\documentclass[11pt,reqno, pdftex]{amsart}

\usepackage{setspace}
\usepackage[parfill]{parskip}
\usepackage[utf8]{inputenc}
\usepackage[margin=3cm]{geometry}
\RequirePackage{natbib}
\RequirePackage[colorlinks,citecolor=blue,urlcolor=blue]{hyperref}
\usepackage{mathtools}
\usepackage{comment}
\usepackage{float}
\usepackage{graphicx}
\usepackage{tikz}
\usetikzlibrary{arrows,positioning,calc, external} 

\tikzset{
    >=stealth',
    punkt/.style={
           rectangle,
           rounded corners,
           draw=black, thick,
           text width=7em,
           minimum height=2em,
           text centered},
    punktl/.style={
           re
           tangle,
           rounded corners,
           draw=black, thick,
           
           text width=7em,
           minimum height=2em,
           text centered},
    pil/.style={
           ->,
           shorten <=4pt,
           shorten >=4pt,},
    pildotted/.style={
           ->,
           shorten <=4pt,
           shorten >=4pt,
  dotted,
  },
  external/system call={pdflatex \tikzexternalcheckshellescape 
                                        -halt-on-error
                                        -interaction=batchmode 
                                        -jobname "\image" "\texsource"
                                        && pdftops -eps "\image.pdf"}
}
\usepackage{subfig}
\usepackage{todonotes}

\newcommand{\farv}[1]{}
\newcommand{\marginnote}[1]{}
\newtheorem{theorem}{Theorem}[section]
\newtheorem{definition}[theorem]{Definition}

\newtheorem{cor}[theorem]{Corollary}
\newtheorem{rem}[theorem]{Remark}
\newtheorem{example}{Example}[section]
\newtheorem{remark}[theorem]{Remark}

\newcommand{\bs}{\boldsymbol}

\usepackage{commath}
\newcommand{\dd}{\dif}
\newcommand{\e}{\mathrm{e}}

\newcommand{\Exp}{\mathbb{E}}
\newcommand{\Prob}{\mathbb{P}}

\newcommand\xqed[1]{%
  \leavevmode\unskip\penalty9999 \hbox{}\nobreak\hfill
  \quad\hbox{#1}}
\newcommand\demo{\xqed{$\circ$}}
\newcommand\demoo{\xqed{$\triangle$}}

\newcommand{\mat}[1]{\boldsymbol{\bs #1}}
\newcommand{\vect}[1]{\boldsymbol{\bs #1}}

\begin{document}

\title[Phase--type representations of stochastic interest rates]{Phase--type representations of stochastic interest rates with applications to life insurance}
\author{Jamaal Ahmad}
\address{Department of Mathematical Sciences, University of Copenhagen, Universitetsparken 5, DK-2100 Copenhagen \O, Denmark.}
\email{jamaal@math.ku.dk}
\author{Mogens Bladt}
\address{Department of Mathematical Sciences, University of Copenhagen, Universitetsparken 5, DK-2100 Copenhagen \O, Denmark.}
\email{bladt@math.ku.dk}

\maketitle
\allowdisplaybreaks

\begin{abstract}
{\farv{blue} The purpose of the present paper is to incorporate stochastic interest rates into a  matrix--approach to multi--state life insurance, where formulas for reserves, moments of future payments and equivalence premiums can be obtained as explicit formulas in terms of product integrals or matrix exponentials. To this end we consider the Markovian interest model, where the rates are piecewise deterministic (or even constant) in the different states of a Markov jump process, and which is shown to integrate naturally into the matrix framework. The discounting factor then becomes the price of a zero--coupon bond which may or may not be correlated with the biometric insurance process. Another nice feature about the Markovian interest model is that the price of the bond coincides with the survival function of a phase--type distributed random variable. This, in particular, allows for calibrating the Markovian interest rate models using a maximum likelihood approach to observed data (prices) or to theoretical models like e.g. a Vasi\v{c}ek model. Due to the denseness of phase--type distributions, we can approximate the price behaviour of any zero--coupon bond with interest rates bounded from below by choosing the number of possible interest rate values sufficiently large. For observed data models with few data points, lower dimensions will usually suffice, while for theoretical models the dimensionality is only a computational issue. 
 }

\vspace{5mm}


\vspace{5mm}


\end{abstract}

\section{Introduction}
This paper considers stochastic interest models, which are state-wise deterministic dependent on an underlying finite state-space Markov process. The spot rate $r(u)$ at time $u$ is assumed to be on the form
\begin{equation}
   r(u) = r_{X(u)}(u),  \label{eq:spot-rate}
\end{equation}  
where $\{ X(u)\}_{u\geq 0}$ denotes a time-inhomogeneous Markov jump process on a $p$--dimensional state--space, and $r_i(u)$, $i=1,...,p$, are deterministic functions. 
{\farv{blue} Assuming an arbitrage free bond market, 
a  zero--coupon bond with terminal date $T$ can then be defined in terms of its prices by \marginnote{Ref $\# 1$\\ intro of $\mathbb{Q}$}
\begin{equation}
  B(t,T) = \Exp^{\mathbb{Q}} \!\left(\left.  \e^{-\int_t^T r_{X(u)}\dd u}\, \right|\, {\mathcal F}(t)  \right)\!, \ \ 0\leq t \leq T , \label{eq:bond}
\end{equation}
where ${\mathcal F}(t) = \sigma (X(u): 0\leq u \leq t)$ is the $\sigma$--algebra generated by $\{ X(u)\}_{u\geq 0}$.\ The expectation is taken under some risk--neutral measure $\mathbb{Q}$ (see, e.g., \cite{bjork2009, elliott-1999}). 
}
{\farv{blue} If all $r_i(u)\geq 0$, a key result of the paper is that, conditionally on $X(t)$, $T\rightarrow B(t,T)$ equals the survival function of an inhomogeneous phase--type distribution. 

In the presence of negative interest rates, this is longer certain since $B(t,T)$ \marginnote{AE} may be larger than one and non--monotone. However, assuming that the negative interest rates are bounded from below by a number $-\rho<0$, we get from \eqref{eq:bond} that
\begin{equation}
    \e^{-\rho(T-t)}B(t,T) = \Exp^{\mathbb{Q}} \!\left(\left.  \e^{-\int_t^T (r_{X(u)}+\rho)\dd u}\, \right|\, {\mathcal F}(t)  \right)\! \label{eq:neg-int}
\end{equation} 
 then equals a survival function of an inhomogeneous phase--type distribution. 
  }

{\farv{blue} The interpretation that the bond prices are (possibly scaled) phase--type survival functions enables us to fit (calibrate) the transition rates of $\{ X(u)\}_{u\geq 0}$ from the observed bond prices by using a maximum likelihood approach. Since phase--type distributions are dense, i.e. can approximate any distribution with a sufficient number of phases,  we may then fit a PH to the observed survival function (equivalent to a histogram) such that all observations (bond prices) are hit. The last point of observation may be considered right censored.
\marginnote{Ref $\#$ 2}
All fitted transition rates are under a risk--neutral measure $\mathbb{Q}$. }

The functional form of the state--wise price of the bond was noted already in \cite[(3.17)]{norberg_2003}, though its relation to phase--type theory was not mentioned, and its potential was not further explored. We also believe that the ``bond price representation'' \eqref{eq:bond} of a phase--type survival function is unknown to the phase--type community.

{\farv{blue}
In the context of multi-state life insurance, modelling stochastic interest rates also play a crucial role. The literature varies from SDE based models, see e.g. \cite{norberg-moller, mollersteffensen, Buchardt2014,banos2020life,asmussensteffensen}, \marginnote{AE} to the finite state--space Markov chain models of \cite{NorbergStokInt, NorbergHigherOrder} on the form \eqref{eq:spot-rate}.\ In the SDE-based methods, one often relies on an independence assumption between interest rates and biometric risk so that available forward rate curves can be used for valuation;\ an exception is \cite{Buchardt2014}, where dependence between interest rates and biometric risk is incorporated. In either case, the SDE-based models do not integrate into classic Thiele and Hattendorf type of results, which limits time-dynamic valuations based on these traditional methods. 
}

The spot rate model \eqref{eq:spot-rate}{\farv{blue}, however,} can be wholly incorporated into Thiele and Hattendorf type of differential equations for reserves and higher order moments, as shown by \cite{NorbergHigherOrder, NorbergStokInt} and further explored in \cite{norberg_2003}.  {\farv{blue} These observations allow for  dependency between interest rates and transitions in life insurance, as well as \marginnote{AE} time-dynamic valuations, without altering the traditional methods}. The latter refers to the model \eqref{eq:spot-rate} as the Markov chain market while \cite{koller} refers to it as Markovian interest intensities. 

{\farv{blue} In this paper, we work with an extended version of the bond prices,
\begin{equation}
   \Exp^{\mathbb{Q}} \!\left(\left. 1\{ X(T) =j \} \e^{-\int_t^T r_{X(u)}\dd u}\, \right|\, {\mathcal F}(t)  \right)\!, \ \ \ j=1,...,p ,  \label{eq:ontheevent}
\end{equation}
which in an insurance context are the discounting factors on the event that the terminal state will be $j$. Providing a matrix--representation for \eqref{eq:ontheevent}, 
we then find how it naturally integrates into the matrix framework of \cite{Bladt2020}. The extension is convenient from a mathematical point of view and also relates to the partial (\cite{Bladt2020}) and retrospective reserves in single states (\cite[Sec. 5E]{norberg1991}). The treatment of the latter, however, is outside the scope of the current paper.   We restate the results of the latter framework in the context of stochastic interest rates. The proofs, and parts of the exposition, will differ from that of \cite{Bladt2020}. 
}

{\farv{blue}  Markov jump processes in finance are often used in connection with regime switching models or where the different states are used to alter the parameters of usually SDE-driven processes. Here transitions can take place under some physical measure and may have a real--world interpretation. The Markov chain model for \marginnote{Ref. $\# 1$} interest rates \eqref{eq:spot-rate} can be thought of as a regime-switching model under a risk-neutral measure, particularly if the interest rates for each state are known a priori. 
}

The Markov jump process approach can approximate bond price modelling in terms of diffusions. Formal constructions have been made in \cite{bharucha-1997,kurtz-1970,kurtz-1978,mijatovic2013continuously}. Since phase--type distributions form a dense class of distributions on the positive reals, this paper will offer an alternative and parsimonious way to approximate any zero--coupon bond (arbitrarily close) by a bond on the form \eqref{eq:bond}.

The paper is organised as follows. Section \ref{sec:summary} introduces some background and notation. 
Bond price modelling using phase--type distribution is developed in Section \ref{sec:discount}. {\farv{blue} In Section \ref{sec:estimation}, we develop estimation of the Markovian interest rate model, both with and without restricted interest rates}, and we provide examples of calibration to diffusion models and real data. In Section \ref{sec:life} we adjust the life--insurance framework of \cite{Bladt2020} to allow for stochastic interest rates of the form \eqref{eq:spot-rate}. It contains examples of how to set up a model using the fitted bond parameters of Section \ref{sec:discount} as well as a matrix--based method for calculating the equivalence premium, either via Newton's method or as an explicit formula.  In Section \ref{sec:numex} we present a numerical example.  For the sake of exposition, the proofs are deferred to Appendix \ref{sec:proofs}.

\section{Background}\label{sec:summary}

\subsection{Notation}\mbox{ }

Unless otherwise stated, row vectors are denoted by bold Greek lowercase letters (e.g., $\vect{\pi}$) and column vectors by bold lowercase Roman letters (e.g., $\vect{v}$). Elements of vectors are denoted by the same unbold, indexed letters (like $\vect{v}=(v_1,...,v_p)^\prime$). 
The vector $\vect{e}_i$ is the column vector which is $1$ at index $i$ and zero otherwise whereas $\vect{e}=(1,1,...,1)^\prime$. 

Matrices are denoted by bold capital letters (Greek or Roman) and their elements by their corresponding lowercase indexed letters (e.g.$\mat{A}=\{ a_{ij} \}$). If $\vect{v}$ is a vector (row or column), then $\mat{\Delta}(\vect{v})$ denotes the diagonal matrix, which has $\vect{v}$ as diagonal. 

\subsection{The product integral}\mbox{ }

Consider a time-inhomogeneous Markov jump process $X = \{X(t)\}_{t\geq 0}$ taking values in a finite state space $E = \{1,\ldots,p\}$, with intensity matrix (functions) $\mat{M}(t) =\{\mu_{ij}(t)\}_{i,j\in E}$. Denote by $\mat{P}(s,t)=\{ p_{ij}(s,t) \}$ the corresponding transition matrix, the elements of which are the transition probabilities $p_{ij}(s,t)=\mathbb{P} (X(t)=j|X(s)=i)$ for $i,j\in E$. The transition matrix $\mat{P}(s,t)$ then satisfies Kolmogorov's forward and backward differential equations,
\begin{align}
\begin{split}
  \frac{\partial}{\partial t}\mat{P}(s,t)&=\mat{P}(s,t)\mat{M}(t),
\quad   \ \ \mat{P}(s,s) = \mat{I}, \\[0.2 cm]
    \frac{\partial}{\partial s}\mat{P}(s,t)&=-\mat{M}(s)\mat{P}(s,t), \quad \mat{P}(t,t)=\mat{I} . \label{eq:Kolmogorov-diff}
  \end{split}
   \end{align}
The solution to \eqref{eq:Kolmogorov-diff}, which in general is not explicitly available, will be denoted by 
 \begin{equation}
  \prod_s^t \left( \mat{I} + \mat{M}(x)\dd x  \right)  \label{eq:prod_int_notation}
\end{equation}
and referred to as the product integral of $\mat{M}(x)$ from $s$ to $t$. This is also true for general matrix functions $\mat{M}(t)$, which satisfy \eqref{eq:Kolmogorov-diff} but are not intensity matrices.

Product integrals have several nice properties. For any $s,t,u\geq 0$, it satisfies the product rule
\begin{equation}
   \prod_s^u (\mat{I} + \mat{M}(x)\dd x) =   \prod_s^t (\mat{I} + \mat{M}(x)\dd x) \prod_t^u (\mat{I} + \mat{M}(x)\dd x) , \label{eq:prod_int_prod_rule}
\end{equation}
which in turn implies that the product integral is invertible with
  \begin{equation}
    \left[ \prod_s^t (\mat{I} + \mat{M}(x)\dd x) \right]^{-1} = \prod_t^s (\mat{I} + \mat{M}(x)\dd x) . \label{eq:prod_int_inv} 
  \end{equation}
  If all $\mat{M}(x)$ commute, then
  \begin{equation}
  \prod_s^t (\mat{I} + \mat{M}(x)\dd x)  = \exp \left( \int_s^t \mat{M}(x)\dd x\right) . \label{eq:prod_int_commute}
  \end{equation}
  In particular, for $\mat{M}(x)\equiv \mat{M}$, we get
  \begin{equation}
  \prod_s^t (\mat{I} + \mat{M}(x)\dd x) = \e^{\mat{M}(t-s)} . \label{eq:prod_int_constant_matrix}
  \end{equation}
If $\mat{A}(x)$ and $\mat{B}(y)$ commute for all $x,y$, then
  \begin{equation}
  \prod_s^t (\mat{I} + (\mat{A}(x)+\mat{B}(x))\dd x ) = \prod_s^t (\mat{I} + \mat{A}(x)\dd x )\prod_s^t (\mat{I} + \mat{B}(x)\dd x ) .
  \label{eq:sum_commute} \end{equation}
In particular, 
\begin{equation}
\e^{-r(t-s)}\prod_s^t (\mat{I} + \mat{A}(x)\dd x )= \prod_s^t (\mat{I} + \left[ \mat{A}(x) - r\mat{I}\right]\dd x ) ,
  \label{eq:sum_commute1} \end{equation}
where $\mat{I}$ denotes the identity matrix.

{\farv{blue}
\begin{remark}\rm 
\marginnote{Ref. $\# 1$}The idea behind the notation of the product integral comes from a Riemann type of construction using step--functions. 
If we approximate $\mat{M}(x)$ by a piecewise constant matrix function taking values 
$\mat{M}(x_i)$ on $[x_i,x_i+\Delta x_i)$ for $s=x_0<x_1<\cdots <x_N=t$ and where $\Delta x_i=x_{i+1}-x_i$, then by \eqref{eq:prod_int_constant_matrix}   the product integral over $[x_i,x_i+\Delta x_i)$ equals the matrix exponential
\[  \e^{\mat{M}(x_i)\Delta x_i} = \mat{I} + \mat{M}(x_i)\Delta x_i + O(\Delta x_i^2) . \]
By letting $\Delta x_i \rightarrow 0$ and using \eqref{eq:prod_int_prod_rule} we then arrive at the notation \eqref{eq:prod_int_notation}. \demoo 
\end{remark}
}

 A valuable formula for computing integrals involving product integrals is the so--called Van--Loan's formula for product integrals (see \cite[Lemma 2]{Bladt2020}), which states that

\begin{equation}
 \small \arraycolsep=1.0pt\def\arraystretch{1.7}\prod_s^t \left( \mat{I}  +  \begin{pmatrix}
\mat{A}(u) & \mat{B}(u) \\
\mat{0} & \mat{C}(u)
\end{pmatrix} \! \dd u\right) = 
\begin{pmatrix}
\displaystyle\prod_s^t (\mat{I} + \mat{A}(u)\dd u) & \quad \displaystyle\int_s^t \prod_s^x (\mat{I} + \mat{A}(u)\dd u)\mat{B}(x)\prod_x^t (\mat{I} + \mat{C}(u)\dd u)\dd x \\ 
\mat{0} &\displaystyle\prod_s^t (\mat{I} + \mat{C}(u)\dd u)
\end{pmatrix}
\label{eqLvan-loan}
\end{equation}
This formula is valid for matrix functions $\mat{A}(x), \mat{B}(x)$ and $\mat{C}(x)$, which are piecewise continuous. The matrices $\mat{A}(x)$ and $\mat{C}(x)$ are square matrices of possibly different dimensions, so $\mat{B}(x)$ is not necessarily a square matrix. 

Let 
\[ \mat{C}(s,t) = \prod_s^t (\mat{I} + \mat{A}(x)\dd x) \otimes \mat{I} , \]
where $\otimes$ denotes the Kronecker product. {\farv{blue}
\marginnote{Ref. $\# 1$} The Kronecker product between a $p_1\times q_1$ matrix  $\mat{A}=\{a_{ij}\} $ and a $p_2\times q_2$ matrix $\mat{B} = \{  b_{ij} \}$ is defined as the $p_1p_2\times q_1 q_2$ matrix 
\[  \mat{A}\otimes \mat{B} = \{ a_{ij}\mat{B} \}_{i=1,...,p_1,j=1,...,q_1} = \{  a_{ij}b_{k\ell}  \} . \]} Using that $(\mat{A}\otimes\mat{B})(\mat{C}\otimes\mat{D})= (\mat{A}\mat{C})\otimes (\mat{B}\mat{D})$, we get
\begin{eqnarray*}
\frac{\partial}{\partial t} \mat{C}(s,t)&=& 
\prod_s^t (\mat{I} + \mat{A}(x)\dd x)\mat{A}(t)\otimes \mat{I} \\
&=& \left( \prod_s^t (\mat{I} + \mat{A}(x)\dd x)\otimes \mat{I}\right)\left(  \mat{A}(t)\otimes \mat{I}  \right)\\
&=&\mat{C}(s,t)\left(  \mat{A}(t)\otimes \mat{I}  \right),
\end{eqnarray*}
and we conclude that 
 \begin{equation}
  \mat{C}(s,t) = \prod_s^t (\mat{I} + (\mat{A}(x)\otimes \mat{I})\dd x) .  \label{eq:Kronecker_prod1}
\end{equation}
A similar argument gives that 
\begin{equation}
\mat{I} \otimes  \prod_s^t (\mat{I} + \mat{A}(x)\dd x) = \prod_s^t (\mat{I} + (\mat{I}\otimes \mat{A}(x))\dd x) .  \label{eq:Kronecker_prod2} 
\end{equation}
 Finally, if $\mat{A}(t)$ and $\mat{B}(t)$ are  Riemann integrable matrix functions of dimensions $q\times q$ and $p\times p$ respectively, then
 \begin{equation}
   \prod_s^t (\mat{I}+ (\mat{A}(x)\oplus \mat{B}(x))\dd x) =  \prod_s^t (\mat{I}+ \mat{A}(x)\dd x)\otimes 
 \prod_s^t (\mat{I}+ \mat{B}(x)\dd x),\label{eq:Kronecker-sum}
 \end{equation}
 where $\oplus$ denotes the Kronecker sum, defined by $\mat{A}(t)\oplus \mat{B}(t) = \mat{A}\otimes\mat{I} + \mat{I}\otimes \mat{B}(t)$, and where the first $\mat{I}$ has the dimension of $\mat{B}(t)$ and the second $\mat{I}$ has the dimension of $\mat{A}(t)$.  To see this, we notice that $\mat{A}(t)\otimes\mat{I}$ and $\mat{I}\otimes \mat{B}(t)$ commute, so by \eqref{eq:sum_commute} we get that 
  \begin{eqnarray*}
  \prod_s^t (\mat{I}+ (\mat{A}(x)\oplus \mat{B}(x))\dd x) &=&  \prod_s^t (\mat{I}+ (\mat{A}(x)\otimes \mat{I})\dd x)  \prod_s^t (\mat{I}+ (\mat{I}\otimes \mat{B}(x))\dd x) \\
  &=&  \left[ \prod_s^t (\mat{I}+ \mat{A}(x)\dd x) \otimes \mat{I} \right]   \left[\mat{I}\otimes \prod_s^t (\mat{I}+ \mat{B}(x)\dd x) \right] \\
  &=& \prod_s^t (\mat{I}+ \mat{A}(x)\dd x) \otimes \prod_s^t (\mat{I}+ \mat{B}(x)\dd x) .
  \end{eqnarray*}
For further details on Kronecker products and sums, we refer to \cite{Graham} 
 
\subsection{Phase--type distributions}\mbox{ }

Consider a (time--inhomogeneous) Markov jump process $\{ Y(t) \}_{t\geq 0}$, where state $p+1$ is absorbing and  $1,...,p$ are transient. The intensity matrix $\mat{M}(x)$ for $\{ Y(t) \}_{t\geq 0}$ is then on the form
 \begin{equation}
  \mat{M}(x) = \begin{pmatrix}
 \mat{T}(x) & \vect{t}(x) \\
 \vect{0} & 0
 \end{pmatrix} ,  \label{eq:PH_complete_matrix}
 \end{equation}
where $\mat{T}(x)$ is a $p \times p$ sub--intensity matrix {\farv{blue} \marginnote{Ref. $\# 1$} consisting of transition rates between transient states, and $\vect{t}(x) = -\mat{T}(x)\vect{e}$} is a column vector of exit rates, i.e.\ rates for jumping to the absorbing state. Then by Van-Loan's formula \eqref{eqLvan-loan}, the transition matrix for $\{ Y(t)\}_{t\geq 0}$ is given by 
\[   \mat{P}(s,t) = \prod_s^t \left( \mat{I}  + \begin{pmatrix}
\mat{T}(u) & \mat{t}(u) \\
\mat{0} & 0 
\end{pmatrix}\!  \dd u\right) =
\begin{pmatrix}
\displaystyle\prod_s^t (\mat{I} + \mat{T}(u)\dd u) & \quad \vect{e} - \displaystyle\prod_s^t (\mat{I} + \mat{T}(u)\dd u)\vect{e} \\ 
\mat{0} &1
\end{pmatrix}   . \]
Hence $\prod_s^t (\mat{I} + \mat{T}(u)\dd u) $ is the matrix which contains the transition probabilities between the transient states from times $s$ to $t$.

We assume that $\Prob (Y(0)=p+1)=0$, and define $\pi_i=\Prob (Y(0)=i)$. Hence $\vect{\pi}=(\pi_1,...,\pi_{p})$ satisfies that $\vect{\pi}\vect{e}=\sum_i \pi_i=1 $, so that $\vect{\pi}$ is the initial distribution for $\{ Y(t)\}_{t\geq 0}$ concentrated on the transient states only. Then  
\begin{equation}
  \left( \Prob(Y(t)=1),\Prob (Y(t)=2),...,\Prob (Y(t)=p)  \right) =\vect{\pi}\prod_0^t (\mat{I} + \mat{T}(u)\dd u)  \label{eq:PH_dist_at_time_t}
\end{equation} 
is a row vector that contains the probabilities of the process being in the different transient states at time $t$. 

Now let 
\[  \tau = \inf \{  t>0 : Y(t) = p+1 \}  \]
  denote the time until absorption. Then from \eqref{eq:PH_dist_at_time_t} we immediately get that
 \begin{equation}
  \mathbb{P} (\tau >t) = \vect{\pi}\prod_0^t (\mat{I} + \mat{T}(u)\dd u) \vect{e} \label{eq:tail-IPH}
\end{equation}
since the right-hand side equals the probability of the process belonging to {\it any } of the transient states by time $t$, i.e., absorption has not yet occurred. Differentiating \eqref{eq:tail-IPH} and using \eqref{eq:Kolmogorov-diff} we see that $\tau$ has a density on the form
\begin{equation}
f_\tau (x) =   \vect{\pi}\prod_0^x (\mat{I} + \mat{T}(u)\dd u) \vect{t}(x) . \label{eq:IPH-density}
\end{equation}

\begin{definition}
The distribution of $\tau$ is called an inhomogeneous phase--type distribution, and we write $\tau \sim \operatorname{IPH}(\vect{\pi},\mat{T}(x))$, where the indexation of $\mat{T}(x)$ is over $x\geq 0$. \end{definition}

 We do not need to specify $\vect{t}(x)$ since it is implicitly given by $\mat{T}(x)$. Indeed, since row sums of intensity matrices (and hence of \eqref{eq:PH_complete_matrix}) are zero, we have that $\vect{t}(x) = -\mat{T}(x)\vect{e}$. If $\mat{T}(x)\equiv \mat{T}$, then we simply write $\tau \sim \operatorname{PH}(\vect{\pi},\mat{T})$. This corresponds to the underlying Markov jump process being time--homogeneous. 

We also notice $\mat{T}(x) + \mat{\Delta}(\vect{t}(x))$ defines an intensity matrix (without the absorbing state).

The class of phase--type distributions (both PH and IPH) is dense (in the sense of weak convergence) in the class of distributions on the positive reals, implying that any distribution with support $\mathbb{R}_+$ may be approximated arbitrarily close by a phase--type distribution. This result is also of considerable practical importance since phase--type distributions can be fitted both to data and distributions using a maximum likelihood approach. 
For the time--homogenous case, PH, see \cite{AsmussenEM} while for IPH we refer to \cite{Albrecher-Bladt-Yslas-2020}.

\section{Phase--type representations of bond prices}\label{sec:discount} 
Consider the stochastic interest rate model of  \eqref{eq:spot-rate}, and let $E = \{1,\ldots,p\}$ denote the state--space of the Markov jump process $X = \{ X(t)\}_{t\geq 0}$ with intensity matrix  $\mat{M}(t)=\{ \mu_{ij}(t) \}_{i,j\in E}$. Let $\vect{r}(t) = \left(r_1(t),\ldots,r_p(t)\right)'$ be the column vector which contains the interest rate functions.

The main result of this section is the following result. 
\begin{theorem}\label{th:main_bond}
For $i,j\in E$, let
\[  d_{ij}(s,t) = \Exp\! \left(  \left. 1\{ X(t)=j\} \exp \left( -\int_s^t r_{X(u)}(u) \dd u \right)\right|  X(s)=i \right)\!, \quad s\leq t. \]
Then the matrix $\mat{D}(s,t)=\{  d_{ij}(s,t) \}_{i,j\in E}$ has the following representation 
\begin{equation}
  \mat{D}(s,t) = \prod_s^t \left(\mat{I} + \left[\mat{M}(u) -\mat{\Delta}(\vect{r}(u)) \right]\!\dd u  \right)  . \label{eq:D-matrix}
\end{equation}
\end{theorem}

\begin{proof}
Conditioning on the state of $s+\dd s$, we get that
\begin{eqnarray*}
\lefteqn{d_{ij}(s,t)}~~\\
&=&(1+\mu_{ii}(s)\dd s) d_{ij}(s+\dd s,t)(1-r_{i}(s)\dd s )+ \sum_{k\neq i} \mu_{ik}(s)\dd s d_{kj}(s+\dd s,t)(1-r_i(s)\dd s) \\
&=&d_{ij}(s+\dd s,t)(1-r_{i}(s)\dd s ) + \mu_{ii}(s)\dd s d_{ij}(s+\dd s,t) + \sum_{k\neq i} \mu_{ik}(s)\dd s d_{kj}(s+\dd s,t)
\end{eqnarray*}
so that
\begin{eqnarray}\label{eq:dij_diff}
-\frac{\partial}{\partial s} d_{ij}(s,t)&=& -r_i(t)d_{ij}(s,t) + \sum_k \mu_{ik}(t)d_{kj}(s,t) .
\end{eqnarray}
In matrix form, this amounts to
\begin{align}\label{eq:discount_diff}
\frac{\partial}{\partial s} \mat{D}(s,t) = -\left( \mat{M}(s) -\mat{\Delta}(\vect{r}(s))  \right)\!\mat{D}(s,t). 
\end{align}
Noting that $\mat{D}(t,t) = \mat{P}(t,t) = \mat{I}$, we hence conclude that \eqref{eq:D-matrix} holds.
\end{proof}
{\farv{blue}  \begin{rem}\rm
\marginnote{Ref. \#1}
The quantities $d_{ij}(s,t)$ in Theorem \ref{th:main_bond} are introduced as $\vect{r}$--deflated transition probabilities in \cite[Appendix 1]{buchardtfurrermoller}, where the authors derive the differential equation \eqref{eq:discount_diff}.\ While they give a martingale--based proof, we provide a probabilistic sample path argument and give a product integral representation.\demoo      
\end{rem}
}
\begin{rem}\rm
Multiplying both sides of \eqref{eq:discount_diff} with $\mat{e}$ from the right, we recover the differential equation for the state-wise discount factors obtained in \cite[(4.4)]{NorbergStokInt}.\demoo
\end{rem}

{\farv{blue} Assume that all $r_i(x)$ are bounded from below, and let
\[   \rho = \max \left( 0, -\min_{i\in E} \inf_{x\geq 0} r_i(x)  \right)  . \]
Then $\rho=0$ if all interest rates are non--negative, and otherwise $-\rho$ provides a lower bound for all of them. Then we have the following result.
\begin{theorem}\label{th:main_bond}
  The price of the zero--coupon bond \eqref{eq:bond} satisfies  
 \begin{equation}
  B(t,T) =   \Exp^{\mathbb{Q}} \left( \left. \exp \left( -\int_t^T r_{X(u)}(u) \dd u \right)\right| X(t) \right) =\vect{e}_{X(t)}^\prime \mat{D}(t,T)\vect{e} . \label{eq:bond1}
\end{equation}
Conditional on $X(t)=i$, 
\[    T\rightarrow \e^{-\rho(T-t)} B(t,T)   \]
is the survival function for an IPH  distributed random variable, $\tau (t)$, with initial distribution $\vect{e}_i^\prime$ and intensity matrices $\mat{M}(x+t)-\mat{\Delta}(\vect{r}(x+t))-\rho\mat{I}$, $x\geq 0$.

In particular, if all interest rates are non--negative, then $\rho=0$ and the price itself, $ T\rightarrow B(t,T)  $ becomes the survival function.
\end{theorem}
\begin{proof} 
The formula  \eqref{eq:bond1} follows directly from the construction of the $\mat{D}(s,t)$ matrix by summing out over $j$ in $d_{ij}(t,T)$, which corresponds to post--multiplying $\mat{D}(t,T)$ by $\vect{e}$.  Next, we notice that
\[   e^{-\rho(T-t)}\prod_s^t \left(\mat{I} + \left[\mat{M}(u) -\mat{\Delta}(\vect{r}(u)) \right]\!\dd u  \right)  = \prod_s^t \left(\mat{I} + \left[\mat{M}(u) -\mat{\Delta}(\vect{r}(u)) - \rho\mat{I} \right]\!\dd u  \right)  , \]
which follows from \eqref{eq:sum_commute1}. 
The matrix $\mat{M}(x)-\mat{\Delta}(\vect{r}(x))-\rho\mat{I}$ is a sub--intensity matrix, which together with the distribution for $X(t)$  defines a phase--type representation $(\vect{\pi}_t,\mat{M}(x+t)-\mat{\Delta}(\vect{r}(x+t)-\rho\mat{I})$, $x\geq 0$ (starting at time $t$).
\end{proof}
The forward rate $f(t,T)$ is defined by 
\[  f(t,T) = -\frac{\partial}{\partial T}\log B(t,T) . \]
Using Theorem \ref{th:main_bond}, we may write
\[   B(t,T) = \e^{\rho(T-t)}\bar{F}_{\tau (t)}(T)   ,\]
where $\bar{F}_{\tau (t)}(T)=1-F_{\tau (t)}(T) $ denotes the survival function for $\tau (t)\sim \mbox{IPH}(\vect{e}_{X(t)}^\prime,\mat{M}(x+t)-\mat{\Delta}(\vect{r}(x+t))-\rho\mat{I})$.
Then  
\[   -\frac{\partial}{\partial T}\log B(t,T) = -\rho + \frac{f_{\tau (t)}(T)}{1-F_{\tau (t)}(T)} , \]
where $f_{\tau (t)}$ denotes the density function for $\tau (t)$. Hence we have proved the following result. 
\begin{cor}
Conditional on $X(t)=i$, the forward rate $f(t,T)$ equals the hazard rate at $T$ for the random variable $\tau (t) \sim \operatorname{PH}(\vect{e}_i,\mat{M}(x+t)-\mat{\Delta}(\vect{r}(T))-\rho\mat{I})$, less $\rho$. i.e.
\begin{equation}
  f(t,T) = \frac{f_{\tau (t)}(T)}{1-F_{\tau (t)}(T)} - \rho
  .  \label{eq:hazard_rate}
\end{equation}
\end{cor}

}

Another immediate consequence of Theorem \ref{th:main_bond} is the following. 

\begin{cor} Assume that all interest rates are non--negative. Then
conditional on $X(t)=i$, the random variable $\tau (t) \sim \operatorname{IPH}(\vect{e}_i^\prime,\mat{M}(t+x)-\mat{\Delta}(\vect{r}(t+x)))$, $x\geq 0$ then has a c.d.f. given by 
\[   F_{\tau (t)} (T)=1-B(t,T) = \Exp^{\mathbb{Q}}\left( \left.   \int_t^T  r_{X(y)}(y)\e^{-\int_t^y r_{X(u)}(u)\dd u}  \dd y\right| X(t) =i \right)   . \] 
\end{cor}
\begin{proof}
This follows from Theorem \ref{th:main_bond} {\farv{blue} with $\rho=0$} and
\begin{eqnarray*}
f_{\tau (t)}(y)&=& -\frac{\partial}{\partial y}B(t,y)=\Exp^{\mathbb{Q}}\left( \left. r_{X(y)}(y)\e^{-\int_t^y r_{X(u)}(u)\dd u} \right| X(t)=i \right) .
\end{eqnarray*} 
Integrating the expression then yields the result. 
\end{proof}

For the case where $t=0$, the above results are reduced to the following.
\begin{cor}
Assume that all interest rates are non--negative.   
Let  $\tau \sim \operatorname{IPH}(\vect{\pi},\mat{M}(x)-\mat{\Delta}(\vect{r}(x)))$ and let
 $\vect{\pi}=(\pi_1,...,\pi_p)^\prime$ denote the (initial) distribution of $X(0)$. Then
 \begin{eqnarray}
    \mathbb{P} (\tau > T ) &=&   \Exp^{\mathbb{Q}} \left(  \exp \left( -\int_0^T r_{X(u)}(u) \dd u \right) \right)
    \label{cor:tail} \\
    F_{\tau (t)} (T)&=& = \Exp^{\mathbb{Q}}\left(   \int_0^T  r_{X(y)}(y)\e^{-\int_0^y r_{X(u)}(u)\dd u}  \dd y \right) \label{cor:cdf} \\
    f(0,T) &=&\frac{f_{\tau}(T)}{1-F_{\tau }(T)}  . \label{cor:hazard}
\end{eqnarray} 

\end{cor}

\begin{rem}\rm
The density $f_\tau (t)$ has the interpretation of being the expected present value of the current interest rate accumulated in a small time interval arround $t$, and $F_\tau (T)$ is the present value of the total accumulated interest rate during $[0,T]$. \demoo
\end{rem}

\begin{example}\rm
Assume that all interest rates are non--negative. 
If $\{ X(t)\}_{t\geq 0}$ is time--homogeneous and $\vect{r}(t)=\vect{r}=(r_1,...,r_p)$, then we also have that
\begin{eqnarray*}
\Exp^{\mathbb{Q}}\left(  \int_0^T  \e^{-\int_0^y r_{X(u)}(u)\dd u}  \dd y\right)&=& \int_0^T \Exp^{\mathbb{Q}}\left(  \e^{-\int_0^y r_{X(u)}\dd u}  \right)\dd y \\
&=& \int_0^T \Prob (\tau >y)\dd y \\
&=& \int_0^T \vect{\pi} \e^{(\mat{M}-\mat{\Delta}(\vect{r}))y}\vect{e} \dd y \\
&=& \vect{\pi}(\mat{M}-\mat{\Delta}(\vect{r}))^{-1}\e^{(\mat{M}-\mat{\Delta}(\vect{r}))T}\vect{e} - \vect{\pi}(\mat{M}-\mat{\Delta}(\vect{r}))^{-1}\vect{e}\\
&=& \mu \left[ 1 - \tilde{\vect{\pi}}\e^{(\mat{M}-\mat{\Delta}(\vect{r}))T}\vect{e}  \right] \\
&=& \mu \Prob (\tilde{\tau}>T) ,
\end{eqnarray*}
where $\mu = \vect{\pi} \left[ -(\mat{M}-\mat{\Delta}(\vect{r}))  \right]^{-1}\vect{e}$ is the expectation of $\tau$, 
\[   \tilde{\vect{\pi}} = \frac{\vect{\pi} \left[ -(\mat{M}-\mat{\Delta}(\vect{r}))  \right]^{-1}}{\vect{\pi} \left[ -(\mat{M}-\mat{\Delta}(\vect{r}))  \right]^{-1}\vect{e}}  \]
is the stationary distribution of a phase--type renewal process with inter--arrivals being $\operatorname{PH}(\vect{\pi}, \mat{M}-\mat{\Delta}(\vect{r}))$, see \cite[Th. 5.3.4]{bladt-nielsen}, and 
 $\tilde{\tau}\sim \operatorname{PH}(\tilde{\vect{\pi}},\mat{M}-\mat{\Delta}(\vect{r})))$. Hence the swap rate $\rho$ can be expressed as
\begin{eqnarray*}
\rho &=& \frac{\Exp^{\mathbb{Q}}\left(  \int_0^T  r_{X(y)}(y)\e^{-\int_0^y r_{X(u)}(u)\dd u}  \dd y\right)}{\Exp^{\mathbb{Q}}\left(  \int_0^T  \e^{-\int_0^y r_{X(u)}(u)\dd u}  \dd y\right)}=\frac{F_\tau(T)}{\mu \Prob (\tilde{\tau}>T) } \\
&=& \frac{1 - \vect{\pi}\e^{(\mat{M}-\mat{\Delta}(\vect{r}))T}\vect{e} }{\vect{\pi}\left[ -(\mat{M}-\mat{\Delta}(\vect{r})) \right]^{-1}\e^{(\mat{M}-\mat{\Delta}(\vect{r}))T}\vect{e} } .
\end{eqnarray*}\demoo
\end{example}

\section{{\farv{blue} Estimation}}\label{sec:estimation}
{\farv{blue} 
Time--homogeneous phase--type distributions or inhomogeneous phase--type distribution where the sub--intensity matrices are on the form 
\[   \mat{T}(x) = \lambda_\theta (x) \mat{T} ,  \]
for some parametric function $\lambda_\theta (x)$, can be estimated in terms of an EM algorithm. 

An observation from a phase--type distribution is hence considered to be the time until a Markov jump process is absorbed, where all transitions and sojourn times in the different states are unobserved. This makes the estimation an incomplete data problem, which an EM algorithm can solve. Essentially the unobserved sufficient statistics (number of jumps between states, total time in then different states) are replaced by their conditional expectations given data and used in the explicit formulas for the maximum likelihood estimators. This updates the parameters, and the procedure is repeated until convergence. Convergence is secured as the likelihood increases in each step. The limit may be a global or only a local maximum.  

Repeated data (absorption times), of course result in the same conditional expectations given their data.
This carries over to weighted data as well, and hence the EM algorithm may efficiently estimate data in histograms. In particular, we may estimate to theoretical distributions by treating their discretised density as a histogram. This provides the link to fitting the intensity matrix of $\{ X(t) \}_{t\geq 0}$ in  \eqref{eq:spot-rate} through bond prices, \eqref{eq:bond} or \eqref{eq:neg-int}, either in terms of observed data or to a theoretical model.

Indeed, consider bond prices $B(0,T_i)$ available at different maturities $T_1,T_2,...,T_n$.\ 
Then according to Theorem \ref{th:main_bond} we have that\marginnote{AE}
  \[  B(0,T_i) = \vect{\pi}\mat{D}(0,T_i)\vect{e} = \e^{\rho T_i}\mathbb{P}(\tau > T_i), \ \ i=1,2,...,n , \]
for some $\rho >0$ and where $\tau\sim\mbox{IPH}(\vect{\pi},\mat{M}(u) -\mat{\Delta}(\vect{r}(u))-\rho\mat{I} )$. Then $\rho$ must satisfy that 
\[   \e^{-\rho T_i} B(0,T_i) \leq 1, \ \ \ i=1,2,...,n.  \]
This can be achieved by choosing
\[    \rho = \max_{i\in \{1,\ldots,n\}} \left(  \frac{\log B(0,T_i)}{T_i}  \right) . \]
In the life--insurance context in Denmark, by regulation the bond prices (discounting factors) must be
computed from discrete forward rates, $f_d(0,T_i)$, published by the Danish Financial Supervisory Authority. Thus
\[   B(0,T_i) = \left( 1 + f_d(0,T_i)  \right)^{-T_i}    \]  
from which
\[    \frac{\log B(0,T_i)}{T_i}  = -\log (1+f_d(0,T_i))   . \]
Hence
\begin{align}\label{eq:rho_formula}
 \rho = \max_i \left(-\log (1+f_d(0,T_i))  \right) = -\min_i \log (1+f_d(0,T_i)).
\end{align}

Hence calibrating to data $B(0,T_i)$, $i=1,...,n$ can be done by fitting PH or IPH distributions to 
$\e^{-\rho T_i}B(0,T_i)$ using an EM algorithm. The possible interest rates can either be picked by the EM algorithm (referred to as unrestricted interest rates), or we can fix the possible rates to values (or functions) of our choice (restricted interest rates). 

In the former case, 
we obtain a maximum likelihood estimate $(\hat{\vect{\pi}},\hat{\mat{T}}(x))$ for the parameters. The estimate for $\mat{M}(x)$ is then readily obtained from 
\[  \hat{\mat{M}}(x) =  \hat{\mat{T}}(x) +  \mat{\Delta}(\vect{t}(x)) . \]
To find the induced interest rates, we also have from Theorem \ref{th:main_bond} that 
\[  \hat{\mat{T}}(x) =  \hat{\mat{M}}(x) - \mat{\Delta}(\vect{r}(x)) - \rho \mat{I}    \]
so we conclude that the estimated exit rates $\vect{t}(x)$ must satisfy
\[  \vect{t}(x) = \vect{r}(x) + \rho \vect{e},  \]
where $\vect{e}$ is the vector of ones. Hence the induced interest rates are given by
\[   \vect{r}(x) = \vect{t}(x) -  \rho \vect{e}  . \]
Neither the transition rates nor the interest rates are unique, but the resulting discount factor (bond price) is invariant under different representations, which is all that matters regarding reserving in the insurance context.

If, in turn, we decide to choose the possible range of interest rates $r_i(x)$ ourselves, then the EM--algorithm is modified not to update the exit rates.} This modification is easily dealt with by simply removing updates of the latter in the original EM algorithm of \cite{AsmussenEM} or \cite{Albrecher-Bladt-Yslas-2020}.\ See Appendix \ref{sec:app-EM} for details. {\farv{blue}
In this case, the exit rates will be fixed at
\[    \vect{t}(x) = \vect{r}(x) + \rho \vect{e}    \]
so
\[  \hat{\mat{M}}(x) =  \hat{\mat{T}}(x) + \mat{\Delta}(\vect{r}(x))  + \rho \mat{I}  .   \]
While the parametrisation of the transition rates may not be unique, the interest rates remain fixed. 
}

{\farv{blue} We now present two examples of fitting to real data and one example to a theoretical model. The estimation is computed using the R--package {\it matrixdist}. }

\begin{example}[{\farv{blue}Fitting to observed bond prices with restricted interest rates}]\label{ex:fitting_bond_prices}\rm
 Bond prices, $B(0,T)$ as of 31/12/2003 (time zero) with maturities $T=1,2,...,30$ years are available from the Danish Financial Supervisory Authority and given by 0.9755051, 0.9434934, 0.9059545, 0.8679149, 0.8251354, 0.7857250, 0.7472528,
 0.7075066, 0.6679984, 0.6286035, 0.5951316, 0.5625969, 0.5310441, 0.5005108,
0.4710280, 0.4448469, 0.4197550, 0.3958013, 0.3728296, 0.3508858, 0.3319907,
0.3140894, 0.2970098, 0.2808430, 0.2654229, 0.2508400, 0.2369349, 0.2237965,
0.2112725, 0.1994495, respectively. 

This corresponds to an empirical survival distribution to which we can then fit phase--type distributions of different dimensions. Regarding the discretisation, we let $0.5+i$, $i=0,...,29$ denote the data points with probability mass $B(i)-B(i+1)$, where $B(0)=1$, and a right censored data point at 30 with probability mass $B(30)=0.1994495$. {\farv{blue} Since all observed bond prices are less than one, we have $\rho = 0$, corresponding to an environment with non-negative interest rates.}  

We used $p=2,3,4,5,10$ and $15$ phases, with state--wise interest rates being $r_i^p=i/(10p)$, $i=1,...,p$ for the different dimensions $p$. Underlying this choice is the assumption that the interest rates fluctuate between $1\%$ and $10\%$, and the $r_i$'s are obtained as the points that divide the interval $[0,0.1]$ into $p$, including the right endpoint. The vectors $\vect{r}^p=(r_1^p,...,r_p^p)^\prime$ will serve as exit rate vectors of the phase--type distributions to be fitted. 

In Figure \ref{Fig:ZCB} (left), we have plotted the phase--type fits to the empirical survival curve for dimensions $p=2,3,4,5$. At dimension 3, we obtain a decent fit and excellent fits for dimensions 4 and 5. The likelihood values for 4 cases are -3.178171, -3.16838,-3.166633, and -3.166182. Further experimentation with dimensions 10 and 15 resulted in likelihoods of -3.165002
and -3.164654, respectively. However, the plots of bond prices and yields are indistinguishable from the plots corresponding to dimension 5, see Figure \ref{Fig:ZCB2}. We can also assess the quality of the fits by plotting the estimated density function vs. the weighted data, shown in Figure \ref{Fig:ZCB3}. Again the plots for dimensions $p=5,10,15$ are almost indistinguishable. 
Therefore, we conclude that dimension 4 or 5 will suffice to {\it approximate} the bond prices.

The estimates of the sub--intensity matrix $\mat{M}-\mat{\Delta}(\vect{r})$ (under a risk neutral measure $\mathbb{Q}$) for dimensions $p=3,4,5$ are given by 
\[\vspace*{-0.1cm}  \small\arraycolsep=2.5pt\def\arraystretch{1.2}
\left(\begin{array}{ccc} -0.13 & 0.1 & 0 \\ 
0 & -0.41 & 0.34 \\ 
0.14 & 0 & -0.24 \\
\end{array}\right),  \arraycolsep=3.0pt\def\arraystretch{1.2}\left( \begin{array}{cccc} -0.25 & 0.22 & 0.01 & 0 \\ 0.14 & -1.11 & 0.75 & 0.18 \\ 0.06 & 0.29 & -0.63 & 0.2 \\ 0.09 & 0.22 & 0.65 & -1.05 \end{array}\right),  \left(\begin{array}{ccccc} -0.26 & 0.02 & 0.06 & 0.07 & 0.08 \\0.07 & -1.68 & 0.69 & 0.23 & 0.65 \\ 0.19 & 0.32 & -1.99 & 0.93 & 0.48 \\ 0.04 & 0.35 & 0.27 & -1.2 & 0.46 \\ 0.07 & 0.82 & 0.07 & 0.8 & -1.85 \end{array} \right) .\] 
To fit the bond prices, the initial distributions of Markov processes were all on the form $(1,0,...,0)$ of appropriate dimension, i.e., initiation in state 1.
\begin{figure}
  \centering
  \includegraphics[scale=0.40]{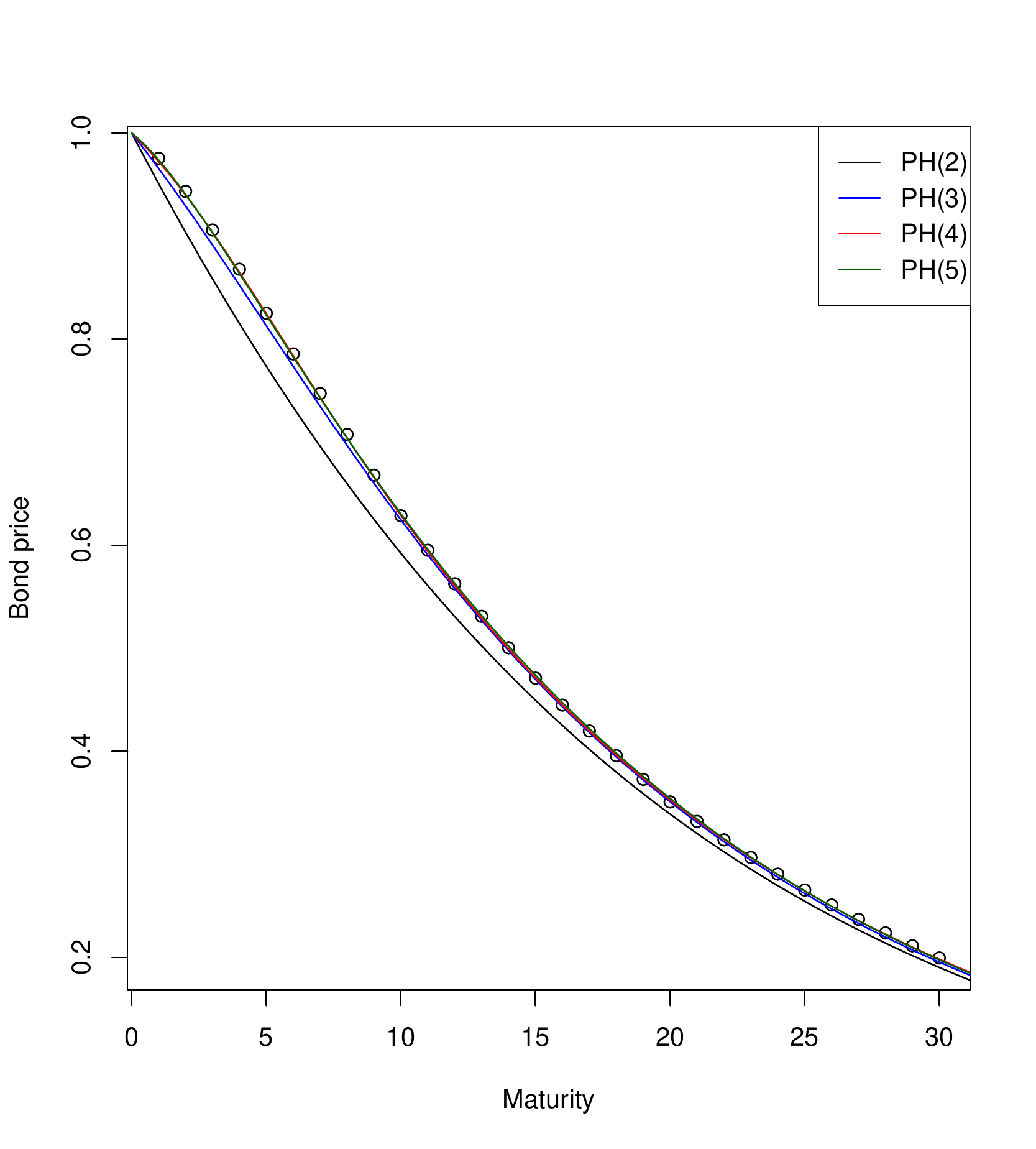}
   \includegraphics[scale=0.40]{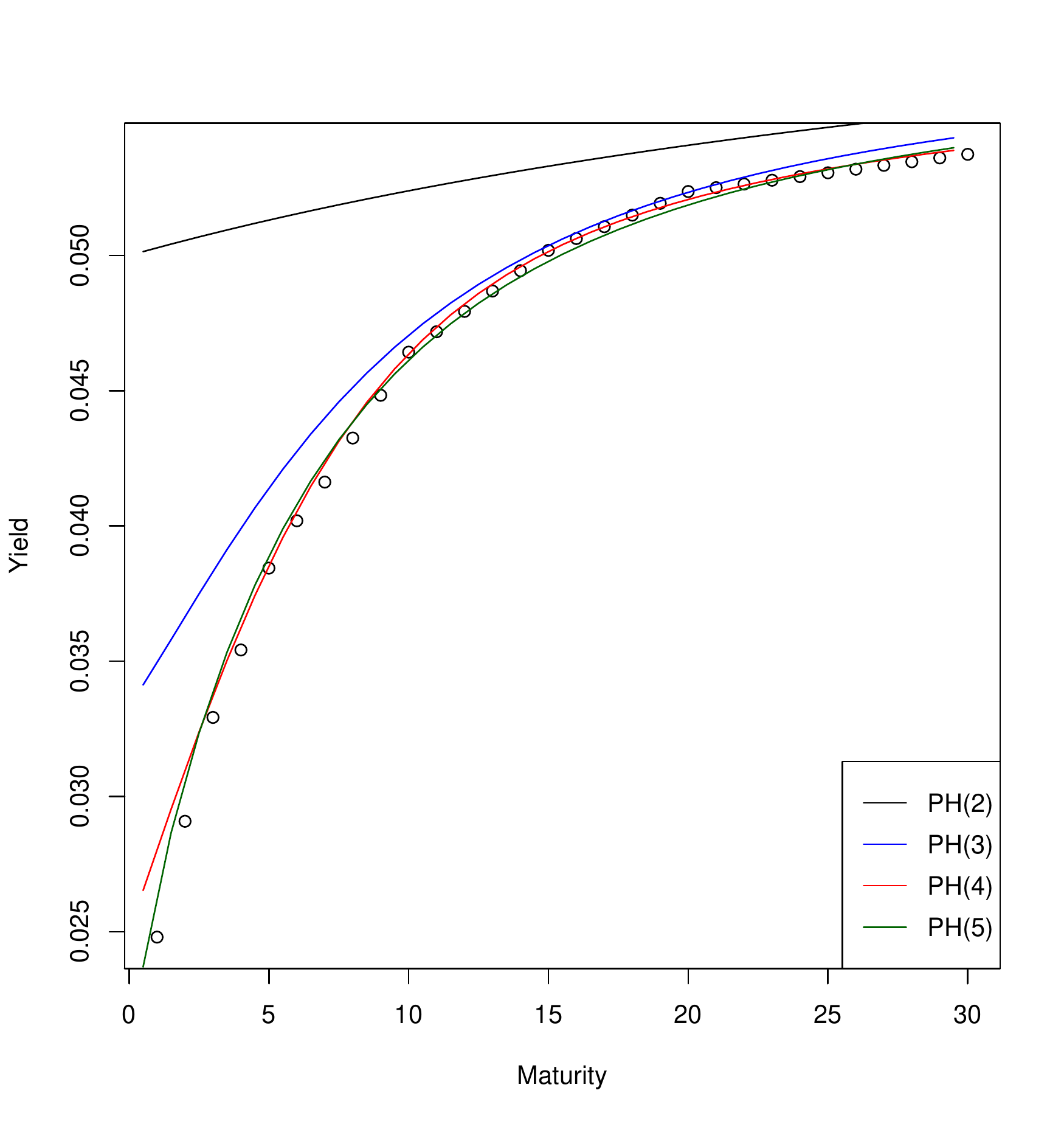}
  \caption{Phase--type fits to Zero-coupon bond prices (left) and corresponding yield curves (right) for dimension $p=2,3,4,5$.}
  \label{Fig:ZCB}
\end{figure}
\begin{figure}
  \centering
  \includegraphics[scale=0.40]{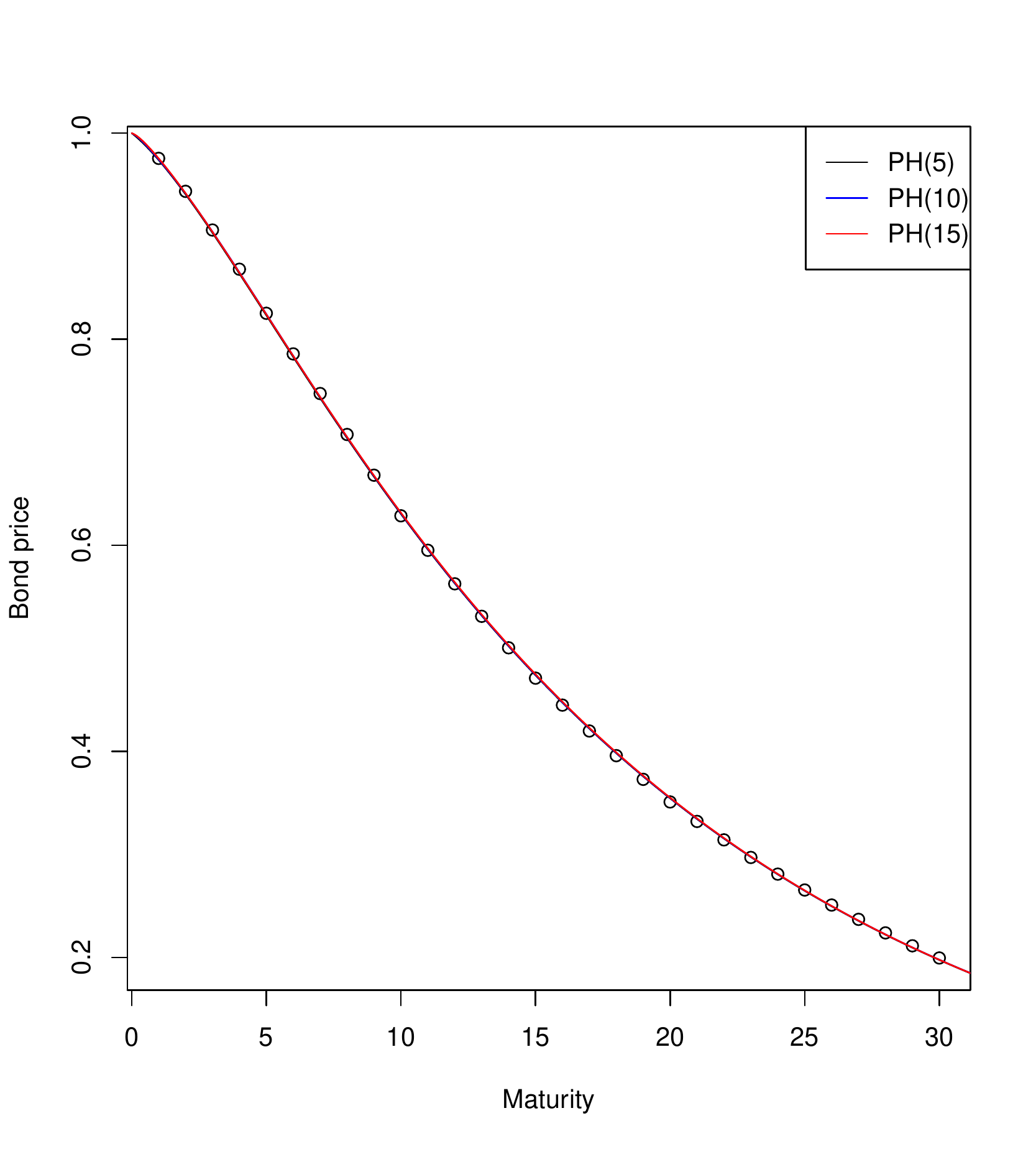}
   \includegraphics[scale=0.40]{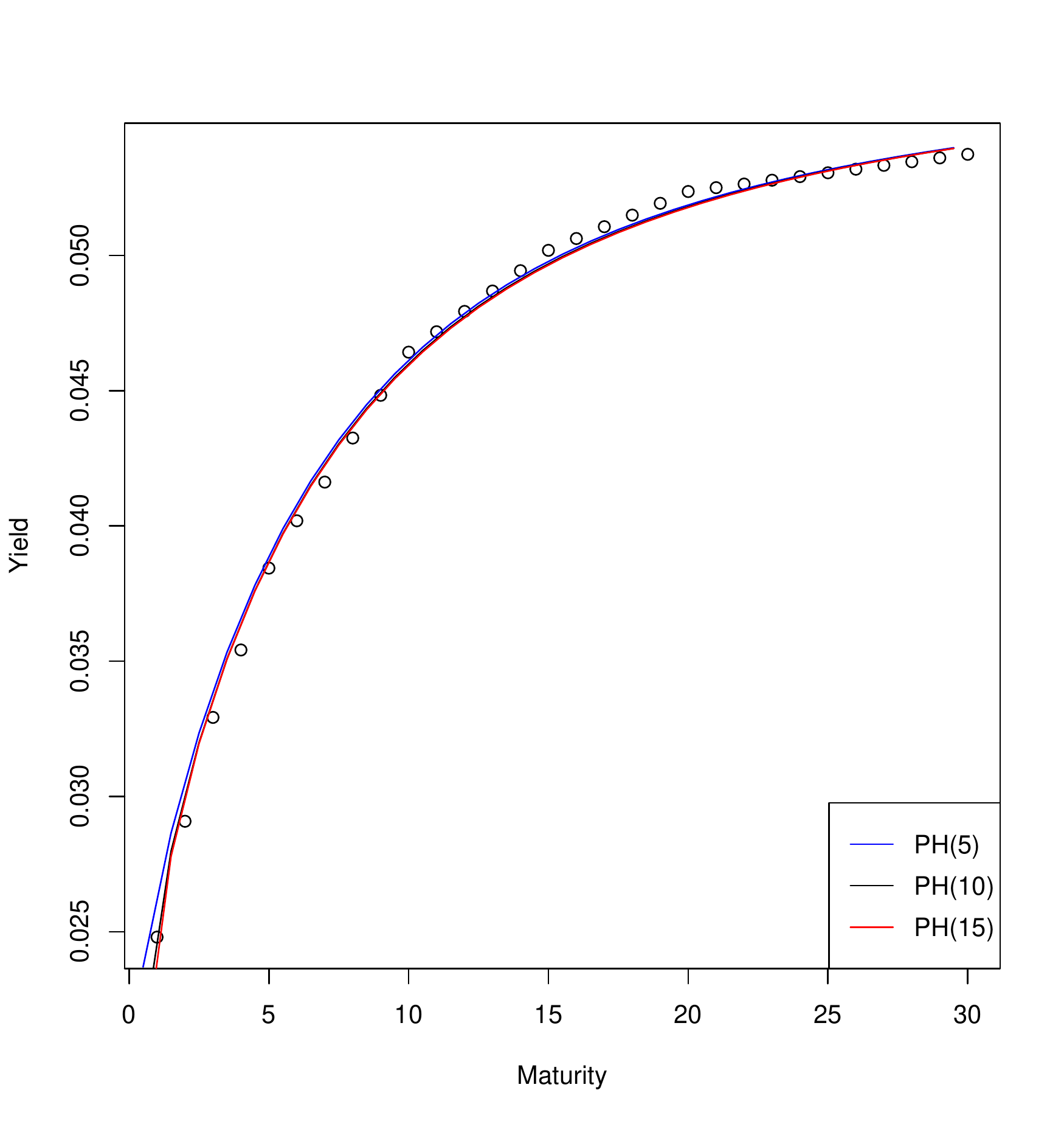}
  \caption{Phase--type fits to Zero-coupon bond prices (left) and corresponding yield curves (right) for dimension $p=5,10,15$.}
  \label{Fig:ZCB2}
\end{figure}
\begin{figure}
  \centering
  \includegraphics[scale=0.40]{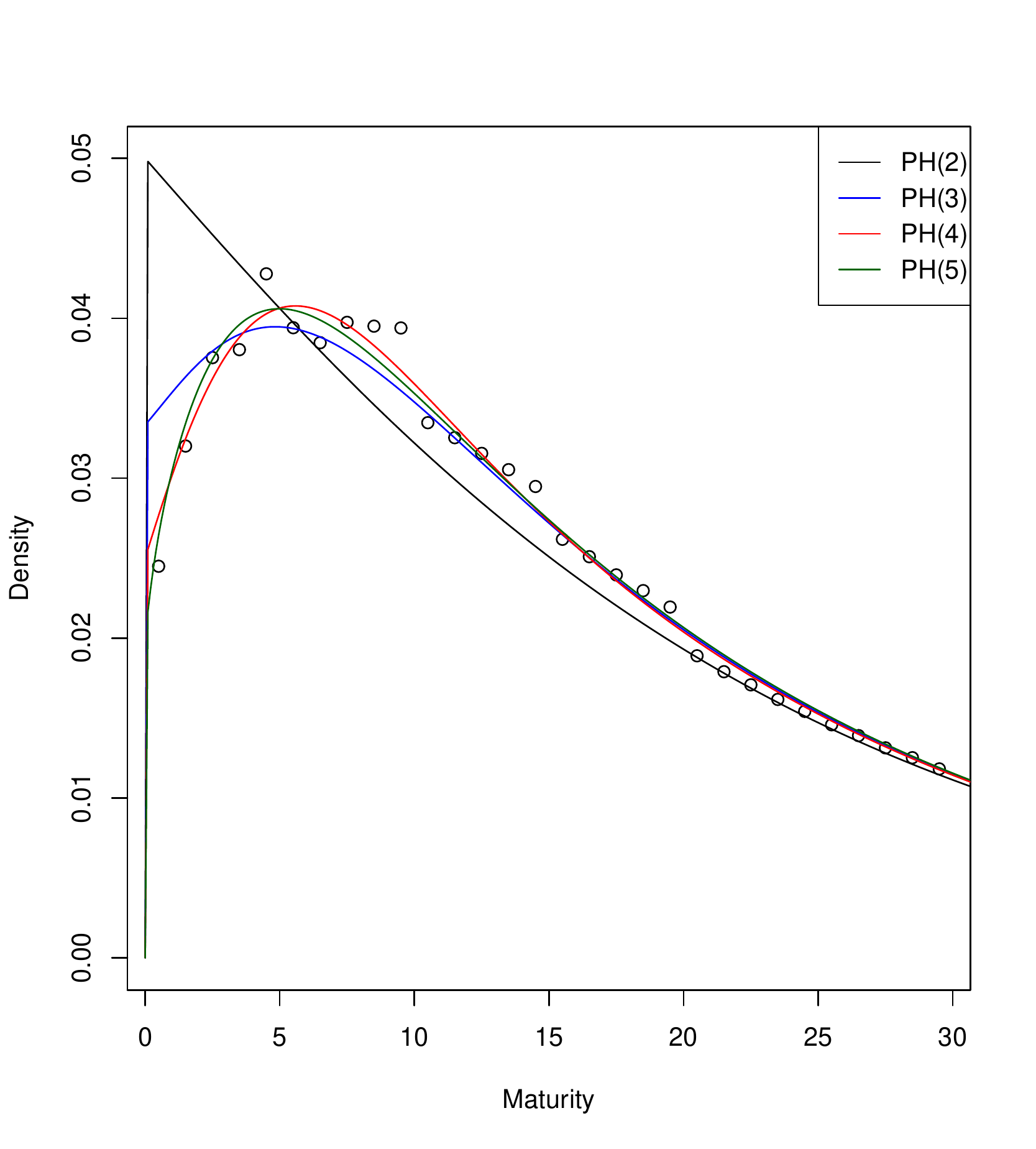}
   \includegraphics[scale=0.40]{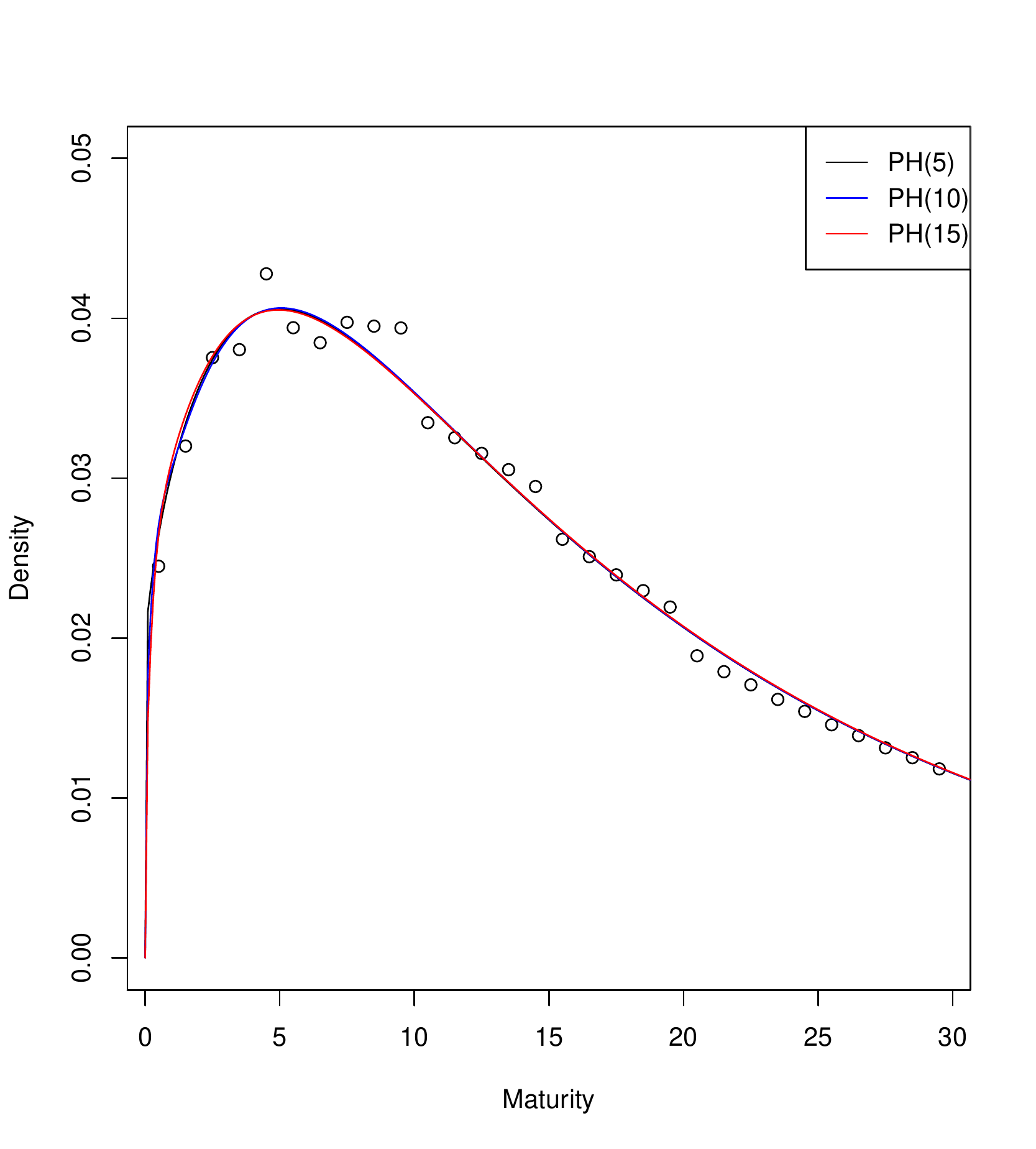}
  \caption{Fitted phase--type densities vs. weighted data for $p=2,3,4,5$ (left) and $p=5,10,15$ (right).}
  \label{Fig:ZCB3}
\end{figure}
\demo
\end{example}

{\farv{blue}
\begin{example}[Fitting to 2019 bond prices with unrestricted interest rates]\rm 
To illustrate the applicability of our methods also in the case of a negative interest rate environment, \marginnote{AE/Ref \#1} we can instead fit to bond prices as of 31/12/2019 from the Danish Financial Supervisory Authority;\ this dataset consists of maturities of $T=1,2,...,120$ years. In this case, we let the EM algorithm choose the necessary positive and negative interest rates. 

The first five years have bond prices above one and given by 1.00231736, 1.00403337, 1.00445679, 1.00382807, and 1.00197787, which reflects the (slightly) negative interest rate environment at the time. From \eqref{eq:rho_formula}, we get $\rho = 0.002314677$ as the exponential factor to down-scale prices to below one. 
\begin{figure}
  \centering
  \includegraphics[scale=0.5]{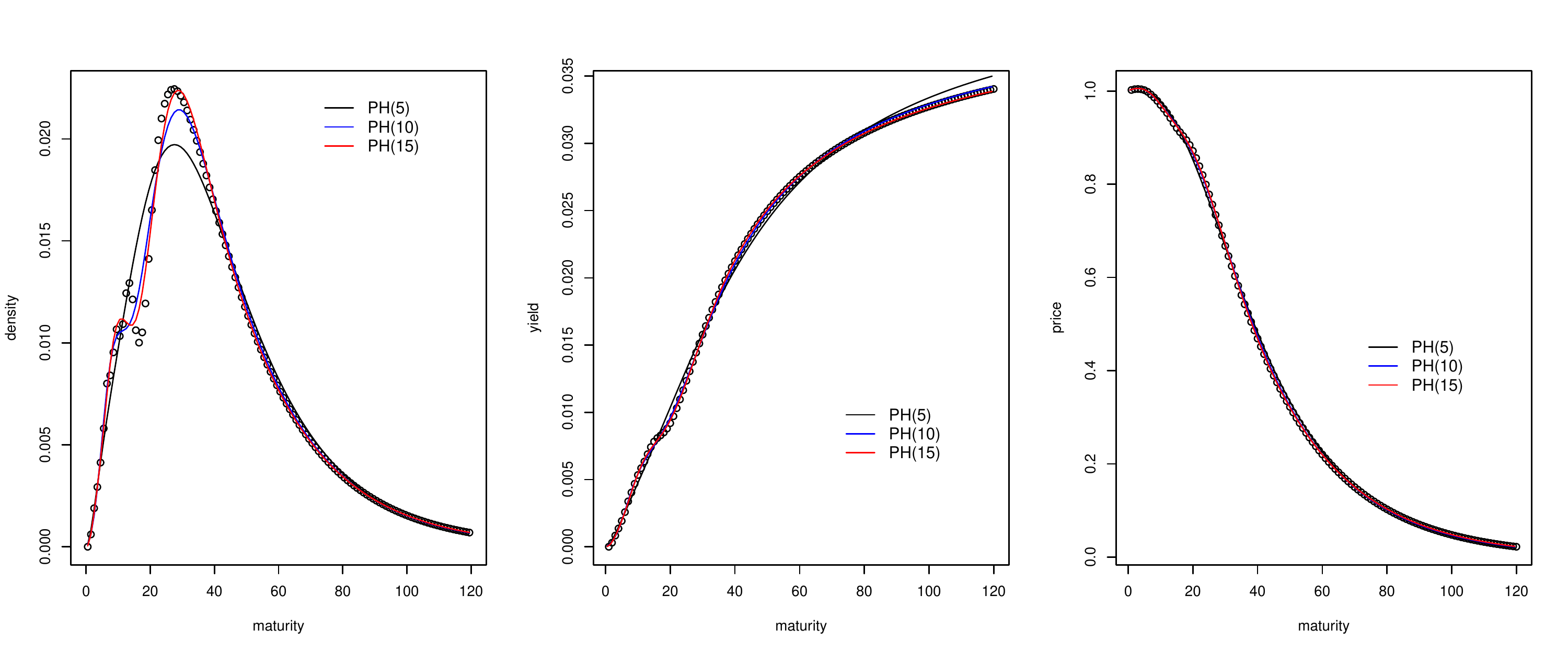}
  \caption{Fitted phase-type densities (left), corresponding yield curves (middle) and bond prices (right) for dimensions $p=5,\!10,\!15$ based on bond price data as of 31/12/2019.} 
  \label{Fig:2019-fit}
\end{figure}
In Figure \ref{Fig:2019-fit}, we show the phase--type fits to the bond prices. We have used the subclass of time--homogeneous Coxian distributions, where initiation is always in state 1, and the only possible transitions are from a state, $i$ say, to the following, $i+1$,  or to exit to the absorbing state. 

If the primary purpose is using the fits as a discounting factor in a life--insurance model, then probably all fits could be used (right plot). If the yield curve fitting is the concern, then only dimensions 10 and 15 seem to catch the appropriate curvature. Regarding the probability density of the phase--type, the 15-dimensional fit is the best. 

To exemplify, we consider the ten dimensional fit. The fitted intensity matrix, $\hat{\mat{M}}$, for $\{ X(u)\}_{u\geq 0}$, is given by
\[  \footnotesize
\left( 
\begin{array}{rrrrrrrrrr}
-0.5212 & 0.5212 & 0.0000 & 0.0000 & 0.0000 & 0.0000 & 0.0000 & 0.0000 & 0.0000 & 0.0000 \\ 
 0.0000 & -0.5212 & 0.5212 & 0.0000 & 0.0000 & 0.0000 & 0.0000 & 0.0000 & 0.0000 & 0.0000 \\ 
 0.0000 & 0.0000 & -0.5185 & 0.5185 & 0.0000 & 0.0000 & 0.0000 & 0.0000 & 0.0000 & 0.0000 \\ 
 0.0000 & 0.0000 & 0.0000 & -0.5161 & 0.5161 & 0.0000 & 0.0000 & 0.0000 & 0.0000 & 0.0000 \\ 
 0.0000 & 0.0000 & 0.0000 & 0.0000 & -0.5152 & 0.5152 & 0.0000 & 0.0000 & 0.0000 & 0.0000 \\ 
 0.0000 & 0.0000 & 0.0000 & 0.0000 & 0.0000 & -0.4664 & 0.4664 & 0.0000 & 0.0000 & 0.0000 \\ 
 0.0000 & 0.0000 & 0.0000 & 0.0000 & 0.0000 & 0.0000 & -0.3099 & 0.3099 & 0.0000 & 0.0000 \\ 
 0.0000 & 0.0000 & 0.0000 & 0.0000 & 0.0000 & 0.0000 & 0.0000 & -0.3099 & 0.3099 & 0.0000 \\ 
 0.0000 & 0.0000 & 0.0000 & 0.0000 & 0.0000 & 0.0000 & 0.0000 & 0.0000 & -0.3099 & 0.3099 \\ 
 0.0000 & 0.0000 & 0.0000 & 0.0000 & 0.0000 & 0.0000 & 0.0000 & 0.0000 & 0.0000 & 0.0000 \\ 
\end{array}
\right)
 \]
The matrix contains six different parameter values. \marginnote{Ref \#1}  The matrix structure is carried over from the
phase--type fit to the (discounted) bond prices. The blocks with the same parameters correspond to Erlang 
blocks, i.e. convolution of exponential distributions with the same parameter.

 The induced (estimated) interest rates (in $\%$) are, respectively,
\[ -\rho, -\rho,   0.03468739,  0.28218594, -\rho,  4.64627655,-\rho
 , -\rho, -\rho,  3.86252219 .   \]
 These should also be counted as parameters. \demo
\end{example}

\begin{example}[Fitting to a two--factor Vasicek model]\rm
In this example we consider the two--factor Vasicek short rate \marginnote{AE} model G2++ (see \cite{Brigo}) with an initial negative interest rate.

Here the bond prices as of time zero are given by
\[   B(0,T) = \exp \left\{-\psi (T)+\frac{1}{2} V^{2}(0, T)\right\},  \]
where 
\[
\begin{aligned} V^{2}(0, T)=& \sum_{i=1}^2 \frac{\sigma_{i}^{2}}{k_{i}^{2}}\left(T-t-B_{k_{i}}(0, T)-\frac{k_{i}}{2} B_{k_{i}}^{2}(0, T)\right) \\
 &+\frac{2 \sigma_{1} \sigma_{2} \sigma_{12}}{k_{1} k_{2}}\left(T-t-B_{k_{1}}(0, T)-B_{k_{2}}(0, T)+B_{k_{1}+k_{2}}(0, T)\right), \end{aligned} \]
\[ B_{k}(0, T)=\frac{1-\e^{-k(T-t)}}{k} \ \ 
\mbox{and}
\ \ \  \psi (T) = \frac{(\theta-r_0)(1+\e^{-k_1 T}) + k_1\theta T}{k_1}.  \]
We chose the same parameters as in \cite{korn}, Fig. 3, apart from the initial interest rate $r_0$, which was set to $-1\%$. 
Hence the parameters are
\[ r_0 = -0.01, k_1 = 0.401, k_2 = 0.178, \sigma_1 = 0.0378, \sigma_2 = 0.0372, \theta = 0.01297, \sigma_{12} = -0.996 \]

\begin{figure}[H]
  \centering
  \includegraphics[scale=0.5]{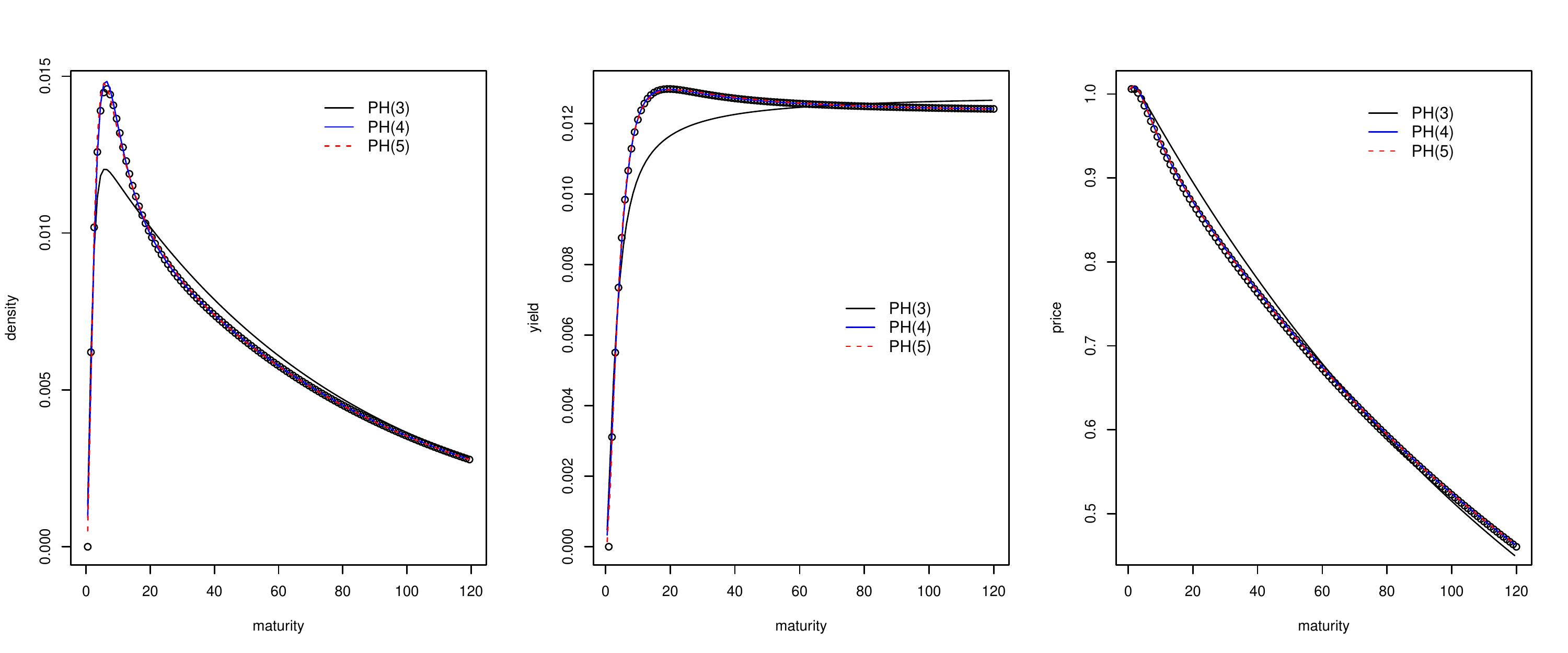}
  \caption{Fitted phase-type densities (left) and corresponding yield curves (middle) and bond prices (right) for dimensions $p=3,\!4,\!5$ based on bond prices from the two-factor Vasicek G2++ model.} 
  \label{Fig:2019-fit}
\end{figure}

We fitted 3,4 and 5 dimensional time--homogeneous phase--type distributions with a Coxian structure to the discounted bond prices $\e^{-\rho T}B(0,T)$.
Here $\rho =  0.005955398$ and the intensity matrix for $\mat{M}$ based on 4 phases is given by
\[    \hat{\mat{M}} = \left( \begin{array}{rrrr}
   -0.17 & 0.17 & 0.00 & 0.00 \\ 
 0.00 & -0.66 & 0.66 & 0.00 \\ 
 0.00 & 0.00 & -0.61 & 0.61 \\ 
 0.00 & 0.00 & 0.00 & 0.00 \\ 
\end{array}\right) \]
with corresponding interest rates $ -\rho,-\rho, 0.078298752  0.006307674$,
while for
5 phases, we get
\[ \hat{\mat{M}} = \left(  
\begin{array}{rrrrr}
-0.65 & 0.65 & 0.00 & 0.00 & 0.00 \\ 
 0.00 & -1.79 & 1.79 & 0.00 & 0.00 \\ 
 0.00 & 0.00 & -1.89 & 1.89 & 0.00 \\ 
 0.00 & 0.00 & 0.00 & -0.12 & 0.12 \\ 
 0.00 & 0.00 & 0.00 & 0.00 & 0.00 \\ 
\end{array}
\right).
  \]
  The corresponding \marginnote{Ref \# 1}(estimated) interest rates are $-\rho,-\rho,-\rho, 0.012793967, 0.006280658$. 
A total of six parameters specify the four-dimensional model, while seven parameters determine the five-dimensional.  \demo
\end{example}

}

\section{Applications to life insurance}\label{sec:life}
In this section, we incorporate the stochastic interest rate model of the previous sections to life insurance valuations. We consider the model introduced by \cite{NorbergHigherOrder, NorbergStokInt} and {\farv{blue} extend their results\marginnote{AE/Ref \# 1} on reserves and higher order moments to so-called partial reserves and higher order moments, that is, corresponding results on events of the terminal state. Partial reserves and moments play important roles when
dealing with so-called retrospective reserves in single states (cf. [5, Section 5.E]), which, however is outside the scope of the present paper. We provide this extension following the matrix approach of \cite{Bladt2020} so that these types of results are extended to allow for stochastic interest rates on the form \eqref{eq:spot-rate}}.\ The extensions of the results of these papers are pointed out in a series of remarks throughout the section.   

\subsection{A Life insurance model with stochastic interest rates}\mbox{}

Let $X=\{ X(t) \}_{t\geq 0}$ be a time--inhomogeneous Markov jump process with a finite state--space $E$ and intensity matrix $\mat{\Lambda}(t)=\{  \lambda_{ij}(t) \}_{i,j\in E}$. Then we define a payment process $\{ B(t) \}_{t\geq 0}$ by 
\begin{align}
\dd B(t) = \sum_{i\in E} 1\{ X(t-)=i\} \bigg(b_{i}(t)\dd t + \sum_{j\in E} b_{ij}(t)\dd N_{ij}(t)\bigg), \label{eq:betalingsfunction}
\end{align}
where $b_i(t)$ are continuous payment rates (negative if premiums) and $b_{ij}(t)$  lump sum payments, which occur according to the counting measure $N_{ij}(t)$. The intensity matrix is decomposed into
\begin{align}
\mat{\Lambda}(t) = \mat{\Lambda}^0(t) + \mat{\Lambda}^1(t), \label{eq:C-D-mat}
\end{align}
where $\mat{\Lambda}^1(t)$ is a non--negative matrix and, consequently, $\mat{\Lambda}^0(t)$ a sub--intensity matrix, i.e. row sums are non--positive.  The counting process is linked to the transitions of $X$ in the following way. Upon transition from $i$ to $j$, $i\neq j$, in $X$ at time $t$, a lump sum payment of $b_{ij}(t)$ will be triggered with probability 
 \begin{equation}
  \frac{\lambda^1_{ij}(t)}{\lambda^0_{ij}(t)+\lambda^1_{ij}(t)} . \label{eq:lump_sum_probability}
\end{equation}
If $i=j$, then $N_{ii}(t)$ denotes an inhomogeneous Poisson process with intensity $\lambda_{ii}(t)$, and a lump sum during a sojourn in state $i$ will then be triggered in $[t,t+\dd t)$ with probability $\lambda^1_{ii}(t)\dd t$.

Finally, we assume that the spot interest rates in state $i$ follow a deterministic function $r_i(t)$. Hence the interest rates follow the model \eqref{eq:spot-rate}. 

\begin{rem}\rm
The classic Markov chain life insurance setting of, e.g., \cite{hoem69, norberg1991}, is the recovered if $r_i(t)\equiv r(t)$, $b_{ii}(t)=0$ and if the probabilities \eqref{eq:lump_sum_probability} are either zero or one. Extending the classic setting to allow for different interest rates in the different states was considered in  \cite{NorbergHigherOrder, NorbergStokInt}, where Thiele type of differential equations for the reserves and higher order moments were derived. \demoo
\end{rem}

For the purpose of computing reserves and higher order moments, \cite[(3.8)--(3.11)]{Bladt2020}, we let $\vect{b}(t)=(b_i(t))_{i\in E}$ denote the vector containing the continuous rates, and define matrices
\begin{align*}
\mat{B}(t) = \left\{b_{ij}(t)\right\}_{i,j\in E}, \ 
\mat{R}(t) = \mat{\Lambda}^1(t)\bullet \mat{B}(t) + \mat{\Delta}(\mat{b}(t)), \
\mat{C}^{(k)}(t) = \mat{\Lambda}^1(t)\bullet \mat{B}^{\bullet k}(t), \quad k\geq 2, 
\end{align*}
where $\mat{\Delta}(\mat{b}(t))$ denotes the diagonal matrix with $\vect{b}(t)$ as diagonal. The operator $\bullet$ denotes Schur (entrywise) matrix product, defined by $\mat{A}\bullet \mat{B} = \{  a_{ij}b_{ij} \}$ for matrices $\mat{A}=\{a_{ij}\}$ and $\mat{B}=\{ b_{ij} \}$.

Hence $\mat{B}(t)$ is the matrix containing the lump payments at transitions and at Poisson arrivals during sojourns, $\mat{R}(t)$ is the matrix whose $ij$'th element is the expected reward accumulated during $[t,t+\dd t)$ upon transition from $i$ to $j$, or during a sojourn in state $i$ if $i=j$. The $\mat{C}^{(k)}(t)$ matrix is more technical to be used when dealing with higher order moments. 

Finally, we let
\begin{align*}
\mat{r}(t) = (r_{i}(t))_{i\in E}.
\end{align*} 
denote the vector of interest rates. 

Now assume that the interest rate process is modelled and fitted using bond prices like in Section \ref{sec:discount}. Accordingly there is a Markov jump process $X_r=\{ X_r(t)\}_{t\geq 0}$ with state--space $E_r=\{1,2,...,p\}$ and intensity matrix 
$\mat{\Lambda}_r(t)=\{ \lambda_{ij}^r(t)\}_{t\geq 0}$, say, such that the corresponding bond prices $B(t,T)$ are given as in Theorem \ref{th:main_bond}. Similarly, we let $X_b=\{ X_b(t)\}_{t\geq 0}$ denote the 
Markov jump process governing the transition between the biometric states with the state--space $E_b=\{1,2,...,q\}$ and intensity matrix $\mat{\Lambda}_b(t)=\{ \lambda_{ij}^b(t)\}_{t\geq 0}$. Hence the Markov jump process appearing in \eqref{eq:betalingsfunction} can be written on the form \begin{align}\label{eq:Xb,Xr}
X(t)=(X_b(t), X_r(t))  
\end{align}
with state--space $E=E_b\times E_r$. 

Hence we need to decide upon an ordering of $E$, which will be lexicographical. This means the for elements $(i,\tilde{i}),(j,\tilde{j})\in E$, 
 \[ (i,\tilde{i})<(j,\tilde{j}) \iff (i-1)p+\tilde{i}<(j-1)p+\tilde{j} .\]
  In other words,  each biometric state $i$ consists of sub--states $(i,1),...,(i,q)$ depending on the state of the underlying Markov process $X_r$, see Figure \ref{Fig:X(t)}. 

\begin{figure}
  \centering
   \includegraphics[scale=0.5]{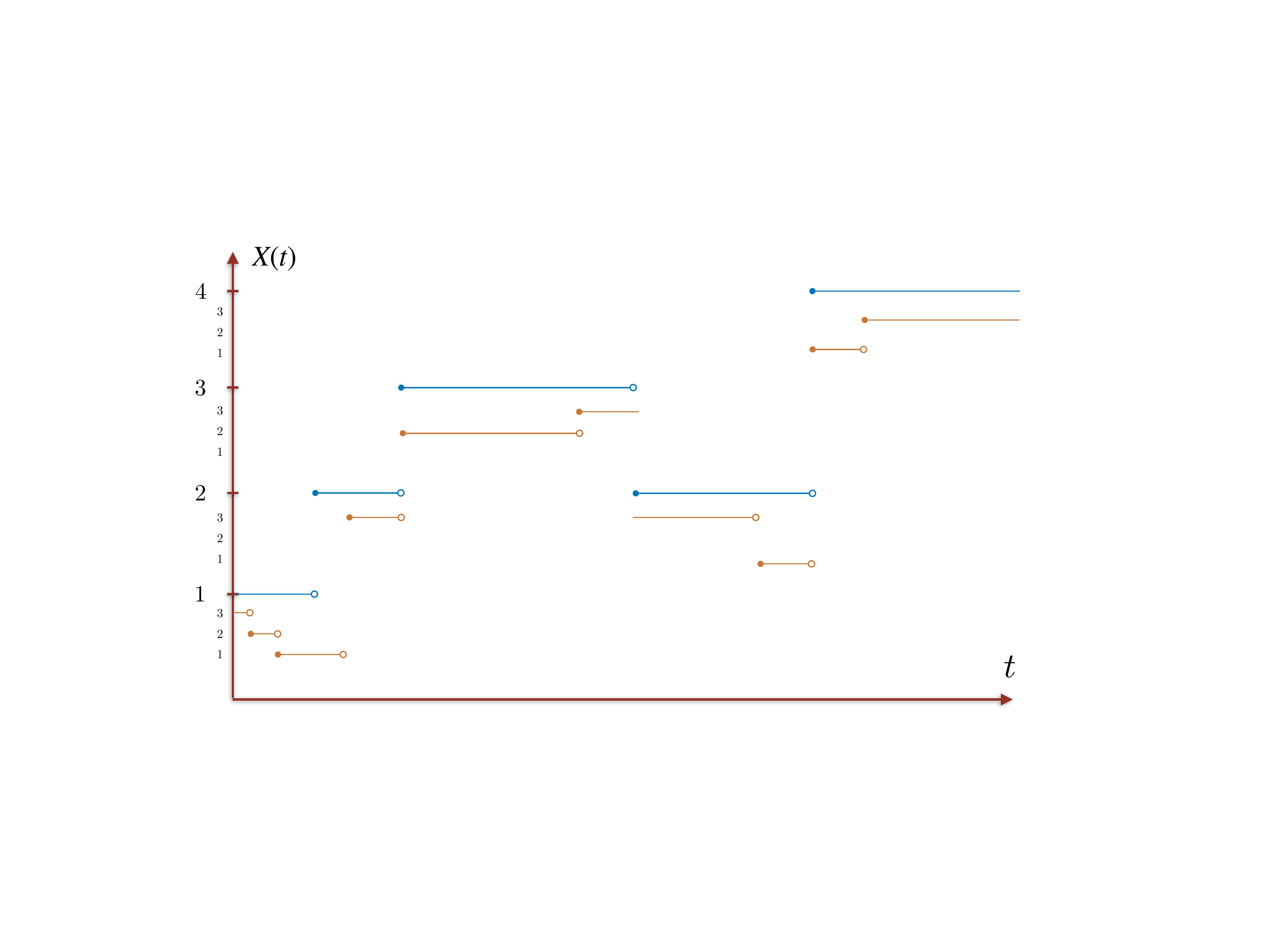}
  \vspace{-4cm}\caption{Lexicographical ordering: for each biometric state (blue), several sub-states (orange) define the underlying interest rate level. }
  \label{Fig:X(t)}
\end{figure}
The processes $X_b$ and $X_r$ may or may not be independent, and the payment processes \eqref{eq:betalingsfunction} likewise may or may not be independent of $X_r$.\ {\farv{blue} In the independent case the processes $X_b$ and $X_r$ are defined on each their state--space, and the common state--space  \marginnote{AE/Ref $\# 2$}will be the product set of the two. If the processes are sharing states, with the possibility of having simultaneous jumps, then we obtain dependency of the processes. Such a case could, e.g. be a rise in the interest rate causing an increased intensity of jumping to surrender or free--policy states (see, e.g., \cite{Buchardt2014}). }

In the following example, we consider the simplifications in the representations when assuming independence.

\begin{example}[Independence]\rm \label{ex:life_1}
If the transition rates of $X$ satisfy, for all  $i,j\in E_b$, $j\neq i$, and $\tilde{i},\tilde{j}\in E_r$, $\tilde{j}\neq \tilde{i}$, 
\begin{align*}
\lambda_{(i,\tilde{i}),(j,\tilde{i})}(t) &= \lambda_{(i,\tilde{j}),(j,\tilde{j})}(t) = \lambda^b_{ij}(t),\\[0.2 cm]
\lambda_{(i,\tilde{i}),(i,\tilde{j})}(t) &= \lambda_{(j,\tilde{i}),(j,\tilde{j})}(t) =: \lambda^r_{\tilde{i}\tilde{j}}(t),
\end{align*}
we have that $X_b$ and $X_r$ are independent. Using the lexicographical ordering, we can, in this case, obtain compact matrix representations in terms of the two processes as follows.  The transition intensity matrix of $X$ is now of the form
\begin{align*}
\mat{\Lambda}(t) = \mat{\Lambda}_b(t) \oplus \mat{\Lambda}_r(t) = \mat{\Lambda}_b(t)\otimes \mat{I}_p + \mat{I}_q \otimes \mat{\Lambda}_r(t) ,
\end{align*}
where $\oplus$ denotes the Kronecker sum, and where $\mat{I}_n$ denotes the identity matrix of dimension $n\times n$. We recall that the Kronecker product, $\otimes$, is defined by 
$\mat{A}\otimes \mat{B} = \{  a_{ij}\mat{B} \}$, where $\mat{A}=\{ a_{ij}\}$. 

The interest rate vector satisfies 
\[   \vect{r}(t) = \vect{e}\otimes (r_1(t),...,r_p(t)),   \]
where $\vect{e}=(1,1,...,1)^\prime$. 

If we further assume that the payment process \eqref{eq:betalingsfunction} is independent of $X_r$, i.e.\ such that the payment functions satisfy, for all  $i,j\in E_b$ and $\tilde{i},\tilde{j}\in E_r$, 
\begin{align*}
b_{(i,\tilde{i})}(t) &= b_{(i,\tilde{j})}(t) =: b^b_i(t), \\[0.2 cm]
b_{(i,\tilde{i}),(j,\tilde{i})}(t) &= b_{(i,\tilde{j}),(j,\tilde{j})}(t) =: b^b_{ij}(t),\\[0.2 cm]
b_{(i,\tilde{i}),(i,\tilde{j})}(t) &= b_{(j,\tilde{i}),(j,\tilde{j})}(t) = 0,
\end{align*}
we have that the payment matrices are on the form 
\begin{align*}
\mat{B}(t) &= \mat{B}^b(t) \otimes \mat{I} \\[0.2 cm]
\vect{b}(t) &= \vect{b}^b(t)\otimes\vect{e} 
\end{align*}
where $$\vect{b}^b(t) = \left(b_1^b(t),\ldots,b_q^b(t)\right)'\qquad \mathrm{and}\qquad  \mat{B}(t) = \left\{b_{ij}^b(t)\right\}_{i,j\in E_b}.$$   
Similarly, we may directly decompose $\mat{\Lambda}_b$:  
\[  \mat{\Lambda}_b(t)= \mat{\Lambda}^1_b(t) \oplus \mat{\Lambda}^0_b(t) \]
such that the decomposition \eqref{eq:C-D-mat} reads
\[  \mat{\Lambda}^1(t) = \mat{\Lambda}_b^1(t)\otimes \mat{I}_q, \ \ \mbox{and} \ \ \mat{\Lambda}^0(t) = \mat{\Lambda}_b^0\otimes \mat{I}_p + \mat{I}_q\otimes \mat{\Lambda}_r(t) = \mat{\Lambda}_b^0(t)\oplus \mat{\Lambda}_r(t), \] 
The conceptual difference in the decomposition of $\mat{\Lambda}^1$ and $\mat{\Lambda}^0$ lies in the absence of lump sum payments upon transition between interest levels.
\demo

\end{example}

\subsection{Reserves}\mbox{ }

We now consider the valuation of the payment process $B$. Introduce the matrix of partial state--wise prospective reserves, 
\begin{align*}
\mat{V}(s,t)&=\left\{ V_{ij}(s,t) \right\}_{i,j\in E}, \\[0.2 cm]
V_{ij}(s,t) &= \Exp\! \left. \left( 1\{ X(t)=j\} \int_s^t \e^{-\int_s^x r_{X(u)}(u)\dd u} \dd B(x)\,  \right|  X(s)=i  \right)\!.
\end{align*}
Due to the stochastic interest rates, this is an extension of \cite{Bladt2020}. With $\mat{D}(s,t)$, introduced in \eqref{eq:D-matrix}, modified to the setup of this section as
\begin{align*}
\mat{D}(s,t) = \prod_s^t \left(\mat{I} + \left[\mat{\Lambda}(u) -\mat{\Delta}(\vect{r}(u)) \right]\!\dd u  \right)\!,
\end{align*}
we have the following result. 
\begin{theorem}\label{th:reserve}
The matrix of partial state-wise prospective reserves $\mat{V}(s,t)$ has the following integral representation: 
\begin{equation}
  \mat{V}(s,t)=\int_s^t \mat{D}(s,x)\mat{R}(x)\mat{P}(x,t)\dd x.  \label{eq:reserve}
\end{equation}
\end{theorem}
\begin{proof}
See Appendix \ref{sec:proofs}.
\end{proof}

The actual computation of the reserves can be effectively executed using the following Van--Loan type of formula, which avoids integration.

\begin{cor}
 $\mat{V}(s,t)$ can be extracted from the relation
\[  
\prod_s^t \left( \mat{I}  + \begin{pmatrix}
\mat{\Lambda}(u) -\mat{\Delta}(\vect{r}(u)) & \mat{R}(u) \\
\mat{0} & \mat{\Lambda}(u)
\end{pmatrix} \! \dd u\right) = \begin{pmatrix}
\mat{D}(s,t) &  \mat{V}(s,t) \\ 
\mat{0} & \mat{P}(s,t)
\end{pmatrix}  .
\]
\end{cor}
Finally, we state and prove Thiele's differential equations for partial reserves with stochastic interest rates.
\begin{theorem}[Thiele]\label{th:thiele}
\[ \frac{\partial}{\partial s}\mat{V}(s,t) =  -\left[ \mat{\Lambda}(s) - \mat{\Delta}(\vect{r}(s))  \right]\!\mat{V}(s,t) - \mat{R}(s)\mat{P}(s,t)   , \]
where $\mat{V}(t,t)=\mat{0}$.\
 For the conventional state--wise prospective reserves, $\mat{V}^{Th}(t)=\mat{V}(t,T)\vect{e}$, this has the form
 \[  \frac{\partial}{\partial t} \mat{V}^{Th}(t) = \mat{\Delta}(\vect{r}(t))\mat{V}^{Th}(t) - \mat{\Lambda}(t)\mat{V}^{Th}(t) - \mat{R}(t)\vect{e} , \]
where $\mat{V}^{Th}(T) = \vect{0}$. 
\end{theorem}
\begin{proof}
See Appendix \ref{sec:proofs}.
\end{proof}

\begin{rem}\rm
Writing out the elements of the differential equation for $\mat{V}^{Th}$, we get for $i\in E$,  
\begin{align*}
\frac{\partial}{\partial t} V^{Th}_{i}(t) &= r_{i}(t)V^{ {\farv{blue}Th}  }_{i}(t) - b_{i}(t)  - \!\!\!\sum_{j\in  E}\!\lambda_{ij}(t)\!\left(b_{ij}(t) + V^{Th}_{j}(t) - V^{Th}_{i}(t)\right)\!,\\[0.2 cm]
V_{i}^{Th}(T) &= 0, 
\end{align*}
which is the differential equation obtained in \cite[(3.2)]{NorbergHigherOrder} in the case of a first-order moment. \demoo 
\end{rem}

\subsection{Higher order moments}\mbox{ }
Consider the matrix of partial state-wise  higher order moments of future payments, given by, for $k\in \mathbb{N}$ (see \cite[(3.6)-(3.7)]{Bladt2020}),
\begin{align*}
\mat{V}^{(k)}(t,T)&=\left\{ V^{(k)}_{ij}(t,T) \right\}_{i,j\in  E}, \\[0.2 cm]
V^{(k)}_{ij}(t,T) &= \Exp\! \left. \left( 1_{ (X(T)=j)} \left(\int_t^T \e^{-\int_t^x r_{X(u)}(u)\dd u} \dd B(x)\right)^{\!k}\,  \right|  X(t)=i  \right)\!,
\end{align*}
and introduce what we shall term the \textit{reduced} partial state-wise higher order moments: 
\begin{align*}
\mat{V}^{(k)}_r(t,T) = \frac{\mat{V}^{(k)}(t,T)}{k!}.
\end{align*}
Since all payment functions and transition rates are deterministic, results for these higher-order moments are now straightforward to obtain by using the undiscounted result, 
\[ 
 \mat{m}^{(k)}_r(t,T)
  =  \int_t^T \mat{P}(t,x)\mat{R}(x)\mat{m}_r^{(k-1)}(x,T)\,\dd x + \sum_{m=2}^k \int_t^T \mat{P}(t,x) 
  \mat{C}_r^{(m)}(x) \mat{m}_r^{(k-m)}(x,T)\,\dd x ,
 \]
 where $\mat{m}^{(k)}_r(t,T)$, $k\in \mathbb{N}$, contains the partial state-wise $k$'th moment, normalised by $k!$, of the undiscounted future payments (see \cite[(7.4)]{Bladt2020}), i.e.\ $\mat{V}^{(k)}_r(s,T)$ with no interest rate. Indeed, rates $b_{i}(t)$ and lump sums $b_{ij}(t)$ must be replaced by the discounted versions with discounting factor, $\exp (-\int_s^t r_{X(u)}\dd u)$ (for fixed $s\leq t$).\ Powers of lumps sums like $b_{ij}(t)^m$, $m\in \mathbb{N}$, are discounted by  $\exp (-m\int_s^t r_{X(u)}\dd u)$.\ Denoting 
 \[  \mat{D}^{(m)}(s,t) = \prod_s^t (\mat{I} +\left[ \mat{\Lambda}(u) - m\mat{\Delta}(\vect{r}(u)) \right]\dd u ), \quad m\in \mathbb{N}, \]
we then obtain the following version of Hattendorff's theorem for partial reserves with stochastic interest rate.
 \begin{theorem}\label{th:moments}
 The matrix of reduced partial state-wise higher order moments satisfies the integral equation, for $k\in \mathbb{N}_0$, 
 \[ 
 \mat{V}^{(k)}_r(t,T)
  =  \int_t^T \mat{D}^{(k)}(t,x)\mat{R}(x)\mat{V}_r^{(k-1)}(x,T)\dd x + \sum_{m=2}^k \int_t^T \mat{D}^{(k)}(t,x) 
  \mat{C}_r^{(m)}(x) \mat{V}_r^{(k-m)}(x,T)\dd x .
 \]
 \end{theorem}
\begin{proof}
See Appendix \ref{sec:proofs}. 
\end{proof}

Defining
\begin{equation*}
\arraycolsep=0.0pt\def\arraystretch{2.2}
\mat{F}_U^{(k)}(x)=
\left(\begin{array}{ccccccc}
\mat{\Lambda}(x)-k\mat{\Delta}(\vect{r}(x))  & \mat{R}(x) & \mat{C}_r^{(2)}(x) &  \cdots & \mat{C}_r^{(k-1)}(x) & \mat{C}_r^{(k)}(x) \\
\mat{0} & \mat{\Lambda}(x)-(k-1)\mat{\Delta}(\vect{r}(x))  & \mat{R}(x) & \cdots & \mat{C}_r^{(k-2)}(x) &\mat{C}_r^{(k-1)}(x) \\
\vdots & \vdots & \vdots &  \vdots \vdots \vdots & \vdots &\vdots  \\
\mat{0} &\mat{0} &\mat{0} & \cdots & \mat{\Lambda}(x)-\mat{\Delta}(\vect{r}(x))  &\mat{R}(x) \\
\mat{0} &\mat{0} &\mat{0} & \cdots  &\mat{0} & \mat{\Lambda}(x) 
\end{array}\right) \label{eq:F-gen-res-1}
\end{equation*}
  we get by Van Loan that 
\begin{equation}
  \prod_t^T (\mat{I} + \mat{F}_U^{(k)}(x)\,\!\dd x) 
=\left(
\arraycolsep=3.0pt\def\arraystretch{1.2}
\begin{array}{lllllll}
* & * & * & * & \cdots &* & \mat{V}_r^{(k)}(t) \\
* & * & * & * & \cdots &* & \mat{V}_r^{(k-1)}(t) \\
* & * & * & * & \cdots &* & \mat{V}_r^{(k-2)}(t) \\
\vdots & \vdots & \vdots &\vdots & \vdots\vdots\vdots & \vdots & \vdots \\
* & * & * & * & \cdots &* & \mat{V}_r^{(1)}(t) \\
* & * & * & * & \cdots &* & \mat{P}(t,T) \\
\end{array} \right) . \label{eq:symbolically}
\end{equation}
From these results, we can derive a number of classical results. Differentiation of  \eqref{eq:symbolically} gives
\begin{eqnarray*}
\lefteqn{\left(
\arraycolsep=3.0pt\def\arraystretch{1.2}
\begin{array}{lllllll}
* & * & * & * & \cdots &* & \frac{\partial}{\partial t} \mat{V}_r^{(k)}(t) \\
* & * & * & * & \cdots &* & \frac{\partial}{\partial t}\mat{V}_r^{(k-1)}(t) \\
* & * & * & * & \cdots &* & \frac{\partial}{\partial t}\mat{V}_r^{(k-2)}(t) \\
\vdots & \vdots & \vdots &\vdots & \vdots\vdots\vdots & \vdots & \vdots \\
* & * & * & * & \cdots &* & \frac{\partial}{\partial t}\mat{V}_r^{(1)}(t) \\
* & * & * & * & \cdots &* & \frac{\partial}{\partial t}\mat{P}(t,T) \\
\end{array} \right) = -\mat{F}_U^{(k)}(t)  \prod_t^T (\mat{I} + \mat{F}_U^{(k)}(x)\,\dd x)}~~~ \\
&=& - \left(\begin{array}{ccccccc}
\mat{\Lambda}(t)-k\mat{\Delta}(\vect{r}(t)) & \mat{R}(t) & \mat{C}_r^{(2)}(t) &  \cdots & \mat{C}_r^{(k-1)}(t) & \mat{C}_r^{(k)}(t) \\
\mat{0} & \mat{\Lambda}(t)-(k-1)\mat{\Delta}(\vect{r}(t)) & \mat{R}(t) & \cdots & \mat{C}_r^{(k-2)}(t) &\mat{C}_r^{(k-1)}(t) \\
\vdots & \vdots & \vdots &  \vdots \vdots \vdots & \vdots &\vdots  \\
\mat{0} &\mat{0} &\mat{0} & \cdots & \mat{\Lambda}(t)-\mat{\Delta}(\vect{r}(t)) &\mat{R}(t) \\
\mat{0} &\mat{0} &\mat{0} & \cdots  &\mat{0} & \mat{\Lambda}(t) 
\end{array}\right) \\
&&\times
\left(
\begin{array}{lllllll}
* & * & * & * & \cdots &* &  \mat{V}_r^{(k)}(t) \\
* & * & * & * & \cdots &* & \mat{V}_r^{(k-1)}(t) \\
* & * & * & * & \cdots &* & \mat{V}_r^{(k-2)}(t) \\
\vdots & \vdots & \vdots &\vdots & \vdots\vdots\vdots & \vdots & \vdots \\
* & * & * & * & \cdots &* & \mat{V}_r^{(1)}(t) \\
* & * & * & * & \cdots &* & \mat{P}(t,T) \\
\end{array} \right)
\end{eqnarray*}
We then obtain the following differential equation by only considering the first row times the last column. 
\begin{theorem}\label{th:reduced-hattendorff}
The matrix of reduced partial state-wise higher order moments satisfies the system of differential equations, for $k\in \mathbb{N}_0$,
 \begin{align*}
   \frac{\partial}{\partial s}\mat{V}_r^{(k)}(t) = \left( k\mat{\Delta}(\vect{r}(u)) -\mat{\Lambda}(t)  \right)\! \mat{V}_r^{(k)}(t) - \mat{R}(t)\mat{V}_r^{(k-1)}(t)
   - \sum_{i=2}^k  \mat{C}_r^{(i)}(t)\mat{V}_r^{(k-i)}(t) , \end{align*} 
   with terminal condition $\mat{V}_r^{(k)}(T)=1_{(k=0)}\mat{I}$.  
\end{theorem}
\begin{rem}\rm
A martingal-based proof for the corresponding (unreduced) state-wise moments, $k!\mat{V}^{(k)}_r(t)\vect{e}$,  can be found in \cite{NorbergHigherOrder}. \demoo
\end{rem}

\begin{example}[Independence continued]\rm
We can continue our decompositions from the independence case of Example \ref{ex:life_1} to reserves and higher-order moments. Indeed, since 
 \begin{eqnarray*}
 \mat{\Lambda}_b(u) \oplus \mat{\Lambda}_r(u) -  k\mat{\Delta}(\vect{e}\otimes \vect{r}(u)) 
 &=& 
  \mat{\Lambda}_b(u)\otimes \mat{I} + \mat{I}\otimes (\mat{\Lambda}_r(u) - k\mat{\Delta}( \vect{r}(u)) )\\
  &=& \mat{\Lambda}_b(u) \oplus (\mat{\Lambda}_r(u) - k\mat{\Delta}( \vect{r}(u)),
   \end{eqnarray*}
 we get from \eqref{eq:Kronecker-sum} that
\begin{eqnarray*}
  \lefteqn{\prod_s^t \left( \mat{I} + \left( \mat{\Lambda}_b(u) \oplus \mat{\Lambda}_r(u) -  k\mat{\Delta}(\vect{e}\otimes \vect{r}(u)) \dd u \right)   \right)}~~\\
  &=& \prod_s^t (\mat{I}+\mat{\Lambda}_b(u)\dd u)\otimes \prod_s^t (\mat{I}+ (\mat{\Lambda}_r(u) - k\mat{\Delta}( \vect{r}(u))\dd u) \\
  &=& \prod_s^t (\mat{I}+\mat{\Lambda}_b(u)\dd u)\otimes \mat{D}^{(k)}(s,t)
\end{eqnarray*}
Thus, each diagonal block element can be computed using these representations when setting up the matrix $\mat{F}_U$ for the computation of these higher order moments. 

In particular, for partial state-wise reserves (i.e.\ $k=1$), we obtain a more direct expression. Assuming that the initial biometric state is $i\in E_b$, the terminal $j\in E_b$ and that the initial distribution of the fitted interest rate phase--type distribution is $\vect{\pi}$. Then
\begin{eqnarray*}
 V_{ij}(t,T) &=&
\lefteqn{(\vect{e}_i^\prime \otimes \vect{\pi})\int_t^T \left( \prod_t^x (\mat{I}+\mat{\Lambda}_b(u)\dd u)\otimes \mat{D}(t,x) \right)\left( \mat{R}(x)\otimes\mat{I} \right)}~~~ \\
&&\times\left( \prod_x^T (\mat{I}+\mat{\Lambda}_b(u)\dd u)\otimes \prod_x^T (\mat{I}+\mat{\Lambda}_r(u)\dd u) \right)\dd x\, (\vect{e}_j\otimes \vect{e}) \\
&=&\int_t^T \vect{\pi}\mat{D}(t,x)\vect{e} \vect{e}_i^\prime \mat{P}_b(t,x)\mat{R}(x)\mat{P}_b(x,T)\vect{e}_j\dd x  \\
&=& \int_t^T\Exp^{\mathbb{Q}}\left(\left.  \e^{-\int_t^T r_{X_r(u)}(u)\dd u } \right| {\mathcal F}(t)  \right) \vect{e}_i^\prime \mat{P}_b(t,x)\mat{R}(x)\mat{P}_b(x,T)\vect{e}_j\dd x, \\
\end{eqnarray*}
which is consistent with similar expressions obtained in \cite{NorbergStokInt}.   \demo
\end{example}

\subsection{Equivalence premium }\label{sec:premiums}\mbox{ }

Assume that $\mat{R}(t)=\mat{R}(t;\theta)$ such that $\theta$ is a parameter of either $\mat{B}(t)$ and/or $\mat{\Delta}(\vect{b}(t))$ only. Hence, $\theta$ could, e.g., be a premium rate in state $1$ or a transition payment between some states. We then write $\mat{V}(t)=\mat{V}(t;\theta)$ so that
\[ \boldsymbol{V}(t;\theta)=\int_{t}^{T} \prod_{t}^{u}\!\left(\boldsymbol{I}+[\boldsymbol{\Lambda}(s)-\mat{\Delta}(\vect{r}(s))] \mathrm{d} s\right) \!\boldsymbol{R}(u;\theta) \prod_{u}^{T}(\boldsymbol{I}+\boldsymbol{\Lambda}(s) \mathrm{d} s) \mathrm{d} u  . \]
If the interest rates satisfy $\mat{\Delta}(\vect{r}(s))\geq \mat{0}$, then $\boldsymbol{\Lambda}(s)-\mat{\Delta}(\vect{r}(s))$ is a sub--intensity matrix, so that $ \prod_{t}^{u}(\boldsymbol{I}+[\boldsymbol{\Lambda}(s)-\mat{\Delta}(\vect{r}(s))] \mathrm{d} s) $ is a sub--probability matrix, i.e. 
\[  0\leq  \prod_{t}^{u}\!\left(\boldsymbol{I}+[\boldsymbol{\Lambda}(s)-\mat{\Delta}(\vect{r}(s))] \mathrm{d} s\right)\!\vect{e} \leq \vect{e}  .  \]
If $\mat{R}(\cdot;\theta)$ is continuously differentiable and $\mat{\Lambda}$ and $\vect{r}$ are continuous, then by Leibniz' integral rule
\[ \frac{\partial}{\partial \theta} \mat{V}(t;\theta) =\int_{t}^{T} \prod_{t}^{u}\!\left(\boldsymbol{I}+[\boldsymbol{\Lambda}(s)-\mat{\Delta}(\vect{r}(s))] \mathrm{d} s\right) \!\frac{\partial}{\partial \theta}\boldsymbol{R}(u;\theta) \prod_{u}^{T}(\boldsymbol{I}+\boldsymbol{\Lambda}(s) \mathrm{d} s) \mathrm{d} u  . \]
Hence we get from the Van Loan formula \eqref{eqLvan-loan}, 
\begin{align}\label{eq:deriv_prod_int}
\prod_t^T \left( \mat{I} + 
\begin{pmatrix}
\mat{\Lambda}(u)-\mat{\Delta}(\mat{r}(u)) & \frac{\partial}{\partial \theta}\mat{R}(u;\theta) \\[1em]
\mat{0} & \mat{\Lambda}(u)
\end{pmatrix}\! \dd u
  \right)  = \begin{pmatrix}
   \mat{D}(t,T) & \frac{\partial}{\partial \theta}\mat{V}(t;\theta) \\[1em]
   \mat{0} & \mat{P}(t,T) 
  \end{pmatrix}. 
\end{align}
\begin{rem}\rm
Similar kinds of derivatives as those of \eqref{eq:deriv_prod_int} are considered in \cite{Kalahsnikov2003}, where differential equations for reserves concerning valuation elements and payments are derived. The formulas presented here may thus be seen as corresponding matrix representations.  \demoo
\end{rem}
If state $i\in E$ is the starting state, we can formulate the equivalence principle by finding the $\theta$ that solves
\[   V_i^{Th}(0;\theta)=\vect{e}_i^\prime \mat{V}(0;\theta)\vect{e} = 0    \]
using Newton's method,  
\[    \theta_{n+1}=\theta_{n} - \frac{ \vect{e}_i^\prime \mat{V}(0;\theta)\vect{e}}{\vect{e}_i^\prime \mat{V}_{\theta}(0;\theta)\vect{e}} ,  \] 
where $\mat{V}_{\theta}$ denotes the partial derivative wrt.\ $\theta$.\ For example, if $\theta$ is a constant premium (rate) such that
\[   \mat{R}_\theta (t;\theta) = \mat{A}(t) ,  \]
i.e.\ a matrix function not depending on $\theta$, then $\mat{V}_{\theta}(t;\theta) = \mat{V}_{\theta}(t)$ will not depend on $\theta$ either, so we conclude that the map $\theta \mapsto V^{Th}_i(t;\theta)$ is linear (for fixed $t$), so that in particular 
\[   V^{Th}_i(0;\theta) = a\theta + b    \]
for some constants $a,b$.\ Then $b$ can be computed from $b=V^{Th}_i(0;0) = \vect{e}_i^\prime \mat{V}(0;0)\vect{e}$ and $a=\vect{e}_i^\prime \mat{V}_{\theta}(0;0)\vect{e}$. 
Hence, Newton's method converges in one iteration, and the $\theta$ which fulfils the equivalence principle is given by 
\begin{align}\label{eq:premium_affine}
\theta = -\frac{ \vect{e}_i^\prime \mat{V}(0;0)\vect{e}}{\vect{e}_i^\prime \mat{V}_{\theta}(0;0)\vect{e}}.
\end{align}
Hence, this formula can compute the equivalence premium if it is assumed to be (piecewise) constant over time, which is often the case in practical examples. However, the formulation in terms of derivatives is usually not seen, with \cite[(3.5)]{Kalahsnikov2003} being one of few exceptions. If the constancy assumption is not satisfied, a parametrised expression in terms of $\theta$ can be calculated by Newton's method.

\subsection{Distributions of future payments based on reduced moments}\label{sec:GC}\mbox{ }

In this section, we briefly comment on the implementation of the Gram--Charlier series for the density and distribution functions based on reduced moments, {\farv{blue} following along the lines of \cite{Bladt2020}; for an approach based on PDEs and integral equations (though not implemented numerically), we refer to \cite[Section 5]{norberg-ano}}. 

The goal is to approximate the distribution of
\[   X = \int_0^T \e^{-\int_0^x r_{X(u)}(u)\dd u} \dd B(x)  \]
using a Gram--Charlier series expansion. In \cite{Bladt2020}, it was shown that under suitable regularity conditions, the density $f$ for $X$ can be approximated by
\[  f(x) \approx f^*(x)\sum_{n=0}^N c_n p_n(x) ,\]
where $f^*$ is a reference density, $p_n(x)$ an orthonormal basis of polynomials for Hilbert space $L^2(f^*)$,
and $c_n = \Exp (p_n (X))$. The reference distribution $f^*$ can be chosen arbitrarily as long as $f/f^*\in L^2(f^*)$. Hence it is advisable to choose $f^*$ as close to $f$ as possible. 

For a given reference density $f^*$, the polynomials 
 \[ 
q_n(x) = 
\left|
\begin{array}{cccc}
a_0  & \cdots & a_{n-1} & 1 \\
a_1 &  \cdots & a_n & x \\
& & \ddots & \\
a_n & \cdots & a_{2n-1} & x^n 
\end{array}
\right| ,
 \]
 where
 \[  a_n = \int_a^b x^n f^*(x)\dd x , \ \ \ n=0,1,...  \]
defines an orthogonal basis for Hilbert space $L^2(f^*)$ with inner product
 \[   \langle g,h \rangle  = \int_a^b g(x)h(x)f^*(x)\dd x . \]
With the Hankel determinants
\[  A_{-1}=1, \ \ \ \ \ A_n = \left|
\begin{array}{cccc}
a_0  & \cdots & a_{n-1} & a_n \\
a_1 &  \cdots & a_n & a_{n+1} \\
& & \ddots & \\
a_n & \cdots & a_{2n-1} & a_{2n}
\end{array}
\right| , n=0,1,..... \]
it can then be shown that
\[  p_n(x) = \frac{q_n(x)}{\sqrt{A_{n-1} A_{n}}}, \ \ n=0,1,...  \]
is an orthonormal basis (ONB) in $L^2(f^*)$. Also, it is immediate that
\[  c_n=\Exp (p_n(X)) = \frac{1}{\sqrt{A_{n-1}A_n}} \left|
\begin{array}{cccc}
a_0  & \cdots & a_{n-1} & 1 \\
a_1 &  \cdots & a_n & \Exp (x) \\
& & \ddots & \\
a_n & \cdots & a_{2n-1} & \Exp (x^n) 
\end{array}
\right|  . \]
If $f^*$ is chosen to be the standard normal distribution, the corresponding polynomials $p_n$ are the (probabilists) Hermite polynomials. While the Hermite polynomials were used in \cite{Bladt2020} up to very high orders, their use in the following example fails already at low orders. This is likely caused by the tail of the normal distribution being too light. We propose a class of reference distributions based on a shifted beta distribution closely related to the Jacobi Polynomials as an alternative. This distribution will have finite support but a much heavier tail. Finite support is usually not a problem in a life insurance context.

Define a reference distribution $f^*$ with support on a finite interval $[a,b]$ by
\[  f^*(x)=\frac{\Gamma (\alpha+\beta+2)}{\Gamma (\alpha+1)\Gamma (\beta+1)} (b-a)^{-\alpha-\beta-1} (b-y)^\alpha  (y-a)^\beta, \ \  x \in [a,b], \ \alpha,\beta >-1 . \]
Thus we need to find an orthonormal basis for $L^2(f^*)$. The starting point is the weight function
\[  w^{\alpha,\beta}(x) = (1-x)^\alpha(1+x)^\beta . \]
The space $L^2(w)$ has an orthogonal basis of Jacobi polynomials given by
\[ q_{n}^{(\alpha, \beta)}(x)=\frac{(\alpha+1)_{n}}{n !} \sum_{k=0}^{n} \frac{(\alpha+\beta+1+n)_{k}(-n)_{k}}{(\alpha+1)_{k} k !}\left(\frac{1-x}{2}\right)^{k} ,  \]
where $(a)_n=a(a+1)\cdots (a+n-1)$ denotes the Pochammer symbol. 

By normalizing the weight function into a density on $[-1,1]$ and then transforming it into a density on $[a,b]$, we obtain an ONB for $f^*$ of polynomials given by 
\[  p_n^{\alpha,\beta}(x)=\sqrt{\frac{n !(2 n+\alpha+\beta+1) (\alpha+\beta+1)_{n}}{ (\alpha+1)_n (\beta+1)_n (\alpha+\beta+1) }  } q_{n}^{(\alpha, \beta)}\left( \frac{2x -a-b}{b-a} \right) .  \]
So for given $a,b$, we need to compute
\begin{eqnarray*}
c_n=\Exp \left( p_{n}^{(\alpha, \beta)}(X)   \right) &=& 
\sqrt{\frac{n !(2 n+\alpha+\beta+1) (\alpha+\beta+1)_{n}}{ (\alpha+1)_n (\beta+1)_n (\alpha+\beta+1) }  } \Exp \left( q_{n}^{(\alpha, \beta)}\left( \frac{2X-a-b}{b-a}\right)   \right) .
\end{eqnarray*}
Here
\begin{eqnarray*}
\lefteqn{\Exp \left( q_{n}^{(\alpha, \beta)}\left( \frac{2X-a-b}{b-a}\right)   \right)}~~\\
 &=& \frac{(\alpha+1)_{n}}{n!} \sum_{k=0}^{n} \frac{(\alpha+\beta+1+n)_{k}(-n)_{k}}{(\alpha+1)_{k}} \frac{1}{k!} \Exp \left( \left(\frac{1- (2X-a-b)/(b-a) }{2}\right)^{k} \right) ,
\end{eqnarray*}
where the inner expectation is computed as
\begin{eqnarray*}
\frac{1}{k!}\Exp \left( \left(\frac{1- (2X-a-b)/(b-a) }{2}\right)^{k} \right) 
&=&\frac{1}{(b-a)^k}\sum_{i=0}^k (-1)^i \frac{b^{k-i}}{(k-i)!} \frac{\Exp (X^i)}{i!} .
\end{eqnarray*}
Finally, the approximation is then given by 
\begin{equation}
  f(x)\approx f^*(x) \sum_{n=0}^N c_n  p_{n}^{(\alpha, \beta)}(x)  . \label{eq:GC_dens_jacobi}
\end{equation}

Concerning the corresponding distribution function, we integrate the above equation to obtain 
\begin{eqnarray*}
  F(y)&\approx& 
  F^*(y) -\frac{b-a}{4} \left(1 - \left(\frac{2y-a-b}{b-a}\right)^2 \right) f^*(y)\\
  &&\hspace{-3mm}\times   \sum_{n=1}^Nc_n\sqrt{ \frac{1}{n}\,\frac{(2+\alpha+\beta)(\alpha+\beta+3)}{(1+\alpha)(1+\beta)(\alpha+\beta+n+1)(\alpha+\beta+n+2)} } p_{n-1}^{(\alpha+1, \beta+1)}\left(\frac{2y-a-b}{b-a}\right).
\end{eqnarray*}
Hence, these formulas can be used to approximate the density and distribution via these Jacobi types of polynomials.

\section{Numerical Example}\label{sec:numex}
We now present a numerical example based on Example \ref{ex:life_1}, where interest rates and biometric risk are assumed independent, 
where we carry over the estimation of interest transition rates from the calibrated bond prices of Section \ref{sec:discount}.

Consider the numerical example of \cite{buchardt2015} as the model for the biometric risk and corresponding life insurance contract. That is, the states of the insured $X_b$ are modelled as a time-inhomogeneous Markov jump process taking values $E_b = \{1,2,3\}$, the three-state disability model depicted in Figure \ref{fig:disability}.
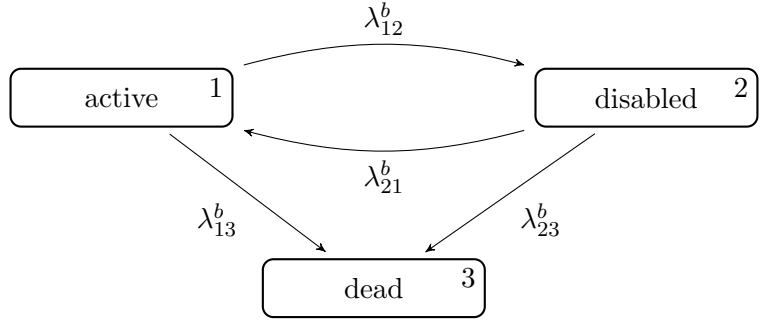
\begin{figure}[htb]
\centering
\begin{tikzpicture}[node distance=2em and 0em]
\node[punkt]                                       (2)    {disabled};
\node[anchor=north east, at=(2.north east)]               {$2$};
\node[punkt, left = 40mm of 2]                     (1)    {active};
\node[anchor=north east, at=(1.north east)]               {$1$};
\node[draw = none, fill = none, left = 20 mm of 2] (test) {};
\node[punkt, below = 20mm of test]                 (3)    {dead};
\node[anchor=north east, at=(3.north east)]               {$3$};    
\path
(1) edge [pil, bend left = 15] node [above]  {$\lambda^b_{12}$} (2)
(2) edge [pil]          node [below right] {$\lambda^b_{23}$} (3)
(1) edge [pil]          node [below left] {$\lambda^b_{13}$} (3)
(2) edge [pil, bend left = 15] node [below]  {$\lambda^b_{21}$} (1)
;
\end{tikzpicture}
\caption{The classic three-state disability model with reactivation}
\label{fig:disability}
\end{figure}

We consider a 40-year-old male today (at time $0$) with a retirement age of 65 and the following life insurance contract: 
\begin{itemize}
\item A disability annuity of rate $1$ while disabled until the retirement of age 65. 
\item A life annuity of rate $1$ while alive until the retirement of age 65.
\item A constant premium rate $\theta$ paid while active until the retirement of age 65, priced under the equivalence principle at time $0$.
\end{itemize}
The maximum contract time is $T = 70$, corresponding to a maximum age of 110 years. The transition rates are given by 
\begin{align*}
\lambda^b_{12}(s) &= \left(0.0004+10^{4.54+0.06(s+40)-10} \right)\!1_{(s\leq 25)}, \\
\lambda_{21}^b(s) &= \left(2.0058 e^{-0.117(s+40)}\right)\!1_{(s\leq 25)}, \\
\lambda_{13}^b(s) &= 0.0005+10^{5.88+0.038(s+40)-10}, \\
\lambda_{23}^b(s) &=  \lambda_{13}^b(s)\!\left(1+1_{(s\leq 25)}\right)\!. 
\end{align*} 
The payment matrices for this product combination corresponds to having $\mat{B}(t) = \mat{\Lambda}^1(t) = \mat{0}$, and
\begin{align*}
\vect{b}(t; \theta) = \begin{cases}
(\theta, 1, 0), \qquad &\text{for} \ t\leq 25 \\[0.2cm]
(1,1,0), &\text{for} \ t>25
\end{cases}.
\end{align*} 
For the stochastic interest rate model, we take the fitted bond prices from Example \ref{ex:fitting_bond_prices} with $p = 4$ phases, so that the interest rates are given as $r(t) = r_{X_r(t)}$, with 
\begin{align*}
\vect{r} = (0.025, 0.050, 0.075, 0.100),
\end{align*} 
and where $X_r$ is a time-homogeneous Markov jump process taking values the finite state space $E_r = \{1,2,3,4\}$ with initial distribution $\vect{\pi} = (1,0,0,0)$ and transition intensity matrix  
\begin{align*}
\mat{\Lambda}_r = \begin{pmatrix}
-0.25 & 0.22 & 0.01 & 0 \\ 0.14 & -1.11 & 0.75 & 0.18 \\ 0.06 & 0.29 & -0.63 & 0.2 \\ 0.09 & 0.22 & 0.65 & -1.05 
\end{pmatrix} + \mat{\Delta}(\vect{r}).
\end{align*}
We then determine the equivalence premium $\theta$ using the method outlined in Section \ref{sec:premiums}. This is explicit on the form \eqref{eq:premium_affine} due to $\vect{b}(t;\cdot)$ being affine (for fixed $t$), and we get $\theta = 0.1583467$.\ This is almost three times lower than the premium rate obtained when pricing with a constant first-order interest rate of $1\%$ as in \cite{buchardt2015}, which makes sense since the present interest rate model always gives interest rates above this level. 

We then calculate moments of up to order 20 of the present value of future payments to approximate its density and distribution function via Gram-Charlier expansions based on the (shifted) Jacobi polynomials, as outlined in Section \ref{sec:GC}. The parameters used in the procedure are shown in Table \ref{tab:Par_GC_BM}, and the resulting density and distribution function are shown in Figure \ref{fig:dens_dist_BM_example}.  
\begin{table}
\centering
\begin{tabular}{c|cccc}
Parameter & $\alpha$ & $\beta$ & $a$ & $b$ \\ \hline Value  
 & 1
 & 0.05 
 & -3
 & 70
\end{tabular}
\caption{Parameters for the Gram-Charlier implementation with (shifted) Jacobi polynomials.}
\label{tab:Par_GC_BM}
\end{table}
\begin{figure}
  \centering
   \includegraphics[scale=0.42]{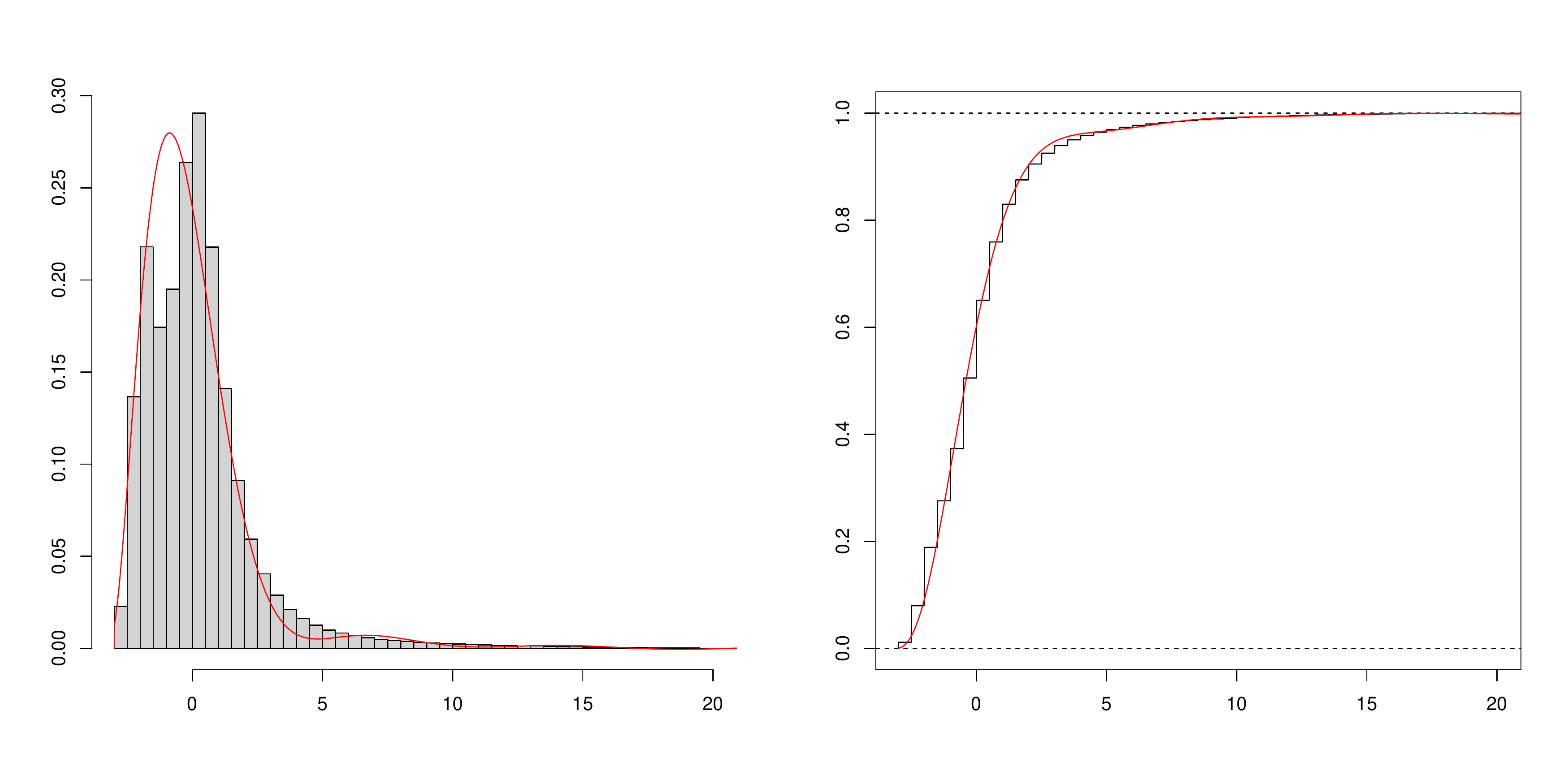}
  \caption{Left:\ Density approximation based on 20 moments and a histogram based on $1,\!000,\!000$ simulations. Right:\ Distribution function approximation based on the same 20 moments and the empirical distribution function from the same simulations.}
  \label{fig:dens_dist_BM_example}
\end{figure}
From the fitted distribution function, one may compute different quantities of interest, e.g., \ quantiles of the present value. In Table \ref{tab:quantiles_BM}, we show various quantiles based on the empirical (simulated) distribution function and the approximated distribution function based on 20 moments. 
\begin{table}[H]
\centering
\begin{tabular}{c|cccc}
  &   \multicolumn{4}{c}{Quantile} \\[0.2 cm]
  &   $95\%$ & $97\%$ & $99\%$ & $99.5\%$ \\ \hline                            
Empirical & 3.51 & 5.51 & 9.51 & 12.01 \\  
Moment-based  & 3.13 & 5.54 & 8.89 & 12.63
\end{tabular}
\caption{Selected quantiles of the present value based on the empirical distribution of $1,\!000,\!000$ simulations and based on the Gram-Charlier approximation based on 20 moments.}
\label{tab:quantiles_BM}
\end{table}

\newpage

\appendix

\section{Modified EM algorithm for phase--type fitting with fixed exit rate vector}\label{sec:app-EM}
First, we consider the case we want to fit a phase--type distribution with parameters $(\vect{\pi},\mat{T}(x))$ to data $y_1,...,y_N$. Here the data are positive real numbers which are thought of as the time until absorption of the underlying Markov process with intensity matrix
\[  \begin{pmatrix}
\mat{T}(x) & \vect{t}(x) \\
\vect{0} & 0 
\end{pmatrix} .   \]
We will assume throughout that $\mat{T}(x)=\mat{T}$, i.e. the Markov process is time--homogeneous. This presents no restriction as long as the interest rate process to be approximated is assumed to be stationary.

If, additionally to the absorption times, we could observe the full trajectories of the Markov process until absorption, then the estimation would be an easy task. In this case, for $i\neq j$
 \begin{equation}
    \hat{\pi}_i =\frac{B_i}{N}, \hat{t}_{ij}=\frac{N_{ij}}{Z_i}, \hat{t}_i=\frac{N_i}{Z_i} \label{eq:mle}
\end{equation}
 whereas $\hat{t}_{ii}=-\hat{t}_i-\sum_{j\neq i}\hat{t}_{ij}$. Here $B_i$ denotes the number of processes starting in state $i$, $N_{ij}$ the number of transitions from $i$ to $j$ in all processes, $N_i$ the number of processes that exits to the absorbing state from state $i$ and $Z_i$ the total time all processes spend in state $i$. 

In the case of incomplete data, where only absorption times are observed, the EM---algorithm can be employed. The idea is to replace the unobserved sufficient statistics $B_i$, $N_{ij}$, $N_i$ and $Z_i$ by the their conditional expectations given data, i.e. $\Exp (B_i| Y=y)$ etc.
The EM--algorithm then alternates between computing these conditional expected values (E--step) and plugging them into \eqref{eq:mle} as a substitute, thereby generating new parameters. 

To perform maximization under the constraint $\vect{t}(x)=\vect{t}=\vect{r}$, we see that this can be achieved simply by removing the update $\hat{t}_i=\frac{N_i}{Z_i}$ from the EM--algorithm, so that the $M$--step amounts to 
 \begin{equation}
    \hat{\pi}_i =\frac{B_i}{N}, \hat{t}_{ij}=\frac{N_{ij}}{Z_i}, \ i\neq j, \ \hat{t}_{ii}=-\hat{r}_i-\sum_{j\neq i}\hat{t}_{ij} .  \label{eq:mle2}
\end{equation}

Fitting a PH distribution to a theoretical distribution is done by approximating the theoretical distribution into a histogram. Hence data will be the discretisation points, and the density values will be the corresponding weights. For further details on the EM--algorithm, we refer to \cite{AsmussenEM} or \cite{Albrecher-Bladt-Yslas-2020}.

\section{Proofs}\label{sec:proofs}
\begin{proof}[Proof of Theorem \ref{th:reserve}]
First, we notice that, for $i,j\in  E$, 
\begin{eqnarray*}
V_{ij}(s,t)&=& \sum_{k\in E} \Exp \left. \left( 1\{ X(t)=j\} \int_s^t1\{ X(x)=k  \} \e^{-\int_s^x r_{X(u)}(u)\dd u} \dd B(x)  \right|  X(s)=i  \right) \\
 &&\hspace{-2cm}= \sum_{k\in E} \int_s^t \Exp \left. \left( 1\{ X(t)=j\} 1\{ X(x)=k  \} \e^{-\int_s^x r_{X(u)}(u)\dd u} \dd B(x)  \right|  X(s)=i  \right) \\
&&\hspace{-2cm}=\sum_{k\in  E} \int_s^t \Exp \left. \left( \Exp \left. \left( 1\{ X(t)=j\} 1\{ X(x)=k  \} \e^{-\int_s^x r_{X(u)}(u)\dd u} \dd B(x)\right| {\mathcal F}_x \right) \right|  X(s)=i  \right) \\
&&\hspace{-2cm}=\sum_{k\in  E} \int_s^t \Exp \left. \left( 1\{ X(x)=k  \} \e^{-\int_s^x r_{X(u)}(u)\dd u}  \Exp \left. \left( 1\{ X(t)=j\}\dd B(x) \right| {\mathcal F}_x \right) \right|  X(s)=i  \right) .\\
\end{eqnarray*}
But on the event $\{ X(x)=k\}$, 
\[ \Exp \left. \left( 1\{ X(t)=j\}\dd B(x) \right| {\mathcal F}_x \right) =
b_k(x)\dd x\  p_{kj}(x,t) + \sum_{\ell\in  E\atop \ell\neq k} b_{k\ell}(x)\nu_{k\ell}(x)\dd x\ p_{\ell j}(x,t)
     \]

so
\begin{eqnarray*}
V_{ij}(s,t)&=&\sum_{k\in  E} \int_s^t p_{kj}(x,t)\Exp \left. \left( 1\{ X(x)=k  \} \e^{-\int_s^x r_{X(u)}(u)\dd u} 
  \right|  X(s)=i  \right)b_k(x)\dd x\\
  &&\hspace{-2cm}+   \sum_{k\in E} \int_s^t \Exp\! \left. \left( 1\{ X(x)=k  \} \e^{-\int_s^x r_{X(u)}(u)\dd u} 
 \left(   \sum_{\ell\in E\atop \ell\neq k} b_{k\ell}(x)\nu_{k\ell}(x)\dd x\ p_{\ell j}(x,t)  \right) \right|  X(s)=i  \right)\\
 &&\hspace{-1.5cm}= \sum_{k\in E} \int_s^t p_{kj}(x,t)\Exp \left. \left( 1\{ X(x)=k  \} \e^{-\int_s^x r_{X(u)}(u)\dd u} 
  \right|  X(s)=i  \right)b_k(x)\dd x\\
  &&\hspace{-2cm}+   \sum_{k\in  E} \int_s^t \left(   \sum_{\ell\in \mathcal{S}\times E\atop \ell\neq k} b_{k\ell}(x)\nu_{k\ell}(x)\ p_{\ell j}(x,t)  \right)  \Exp \left. \left( 1\{ X(x)=k  \} \e^{-\int_s^x r_{X(u)}(u)\dd u} 
 \right|  X(s)=i  \right)  \dd x \\
 &&\hspace{-1.5cm}= \int_s^t  \sum_{k\in  E}  D_{ik}(s,x)b_k(x)p_{kj}(x,t)\dd x + \int_s^t \sum_{k,\ell\in E \atop \ell \neq k } D_{ik}(s,x)b_{k\ell}(x)\nu_{k\ell}(x) p_{\ell j}(x,t)\dd x .
\end{eqnarray*}
In matrix form this amounts to \eqref{eq:reserve}.
\end{proof}
\begin{proof}[Proof of Theorem \ref{th:thiele}]
Using that the product integral satisfies Kolmogorov's forward and backward equations, we get that
\begin{eqnarray*}
 \begin{pmatrix}
\frac{\partial}{\partial s} \mat{D}(s,t) &  \frac{\partial}{\partial s}\mat{V}(s,t) \\ 
\mat{0} & \frac{\partial}{\partial s}\mat{P}(s,t)
\end{pmatrix} &=& \frac{\partial}{\partial s} \prod_s^t \left( \mat{I}  + \begin{pmatrix}
\mat{\Lambda}(u) -\mat{\Delta}(\vect{r}(u)) & \mat{R}(u) \\
\mat{0} & \mat{\Lambda}(u)
\end{pmatrix}  \dd u\right) \\
&=& - \begin{pmatrix}
\mat{\Lambda}(s) -\mat{\Delta}(\vect{r}(s)) & \mat{R}(s) \\
\mat{0} & \mat{\Lambda}(s)
\end{pmatrix}  \begin{pmatrix}
\mat{D}(s,t) &  \mat{V}(s,t) \\ 
\mat{0} & \mat{P}(s,t)
\end{pmatrix},
 \end{eqnarray*} 
from which Thiele's differential equation can be pulled out from the upper right corner of each side of the equation.
 \end{proof}

\begin{proof}[Proof of Theorem \ref{th:moments}]

Write
\begin{eqnarray}
 \lefteqn{ \left( \int_t^T\e^{-\int_t^x r_{X(u)}\dd u} \dd A(x)  \right)^k}~~ \nonumber \\ &=& \int_t^T\cdots \int_t^T \e^{-\int_t^{x_1} r_{X(u)}\dd u}\cdots \e^{-\int_t^{x_k} r_{X(u)}\dd u} \dd A(x_k)\cdots \dd A(x_1) . \label{app:proof1}
  \end{eqnarray}

Now 
\[  \dd A (t) = b^{X(t)}\dd t + b^{X(t-)j}(t)\dd N^{X(t-)j}(t) , \]
and assume that $s\in [t,T]$ is a point of increase for the counting process $x\rightarrow N^{ab}(x)$ which trigger lump sum payments. 
Then in the computation of the above integral, there will be jump contributions at time $s$, where any number $m\in \{1,2,...,k\} $ of the variables $x_1,...,x_k$ may be equal to $s$, say $x_{i_1}=\cdots x_{i_m}=s$. We can pick $m$ out of the $k$ variables in ${k \choose m}=k!/(m! (k-m)!)$ ways. If $m$ variables coincide at the jump time $s$, then a contribution of $b^{ab}(s)^m$ is added. 
Hence only looking at jump coincidences, i.e. $m\geq 2$, the contribution to the integral \eqref{app:proof1} is
\begin{eqnarray*}
 \lefteqn{ \sum_{m=2}^k {k \choose m}\int_t^T \e^{-m \int_t^{s} r_{X(u)}\dd u}  b^{ab}(s)^m }~~\\
 && \times \Bigg( \int_{s}^T \cdots \int_s^T  \e^{-\int_t^{x_{m+1}} r_{X(u)}\dd u}\cdots \e^{-\int_t^{x_k} r_{X(u)}\dd u} \dd A(x_k)\cdots \dd A(x_{m+1}) \Bigg) \dd N^{ab}(s)  .
\end{eqnarray*}

Indeed, since there are precisely $m$ coincidences, the remaining integrals must start from $s+=s$; otherwise, the integration intervals would contain $s$ as well. Changing the lower limits of the integrals appearing in the exponentials, we can further rewrite the expression as

\begin{equation}\label{app:proof2}
\begin{split}
\sum_{m=2}^k{k \choose m}\int_t^T \e^{-k \int_t^{s} r_{X(u)}\dd u}  b^{ab}(s)^m & \\
 &\hspace{-5cm}\times \Bigg( \int_{s}^T \cdots \int_s^T  \e^{-\int_s^{x_{m+1}} r_{X(u)}\dd u}\cdots \e^{-\int_s^{x_k} r_{X(u)}\dd u} \dd A(x_k)\cdots \dd A(x_{m+1}) \Bigg) \dd N^{ab}(s)  .
\end{split}
\end{equation}

Taking conditional expectation $\Exp ( 1\{X(T)=j\} \cdot | X(t)=i)$ of \eqref{app:proof2}, we get 
\begin{eqnarray}
 \lefteqn{ \sum_{m=2}^k{k \choose m}\int_t^T \Exp \bigg(1\{ X(T)=j \} \e^{-k \int_t^{s} r_{X(u)}\dd u}  b^{ab}(s)^m }~~\nonumber \\
 &&\hspace{-1.0cm} \times \Bigg( \int_{s}^T \cdots \int_s^T  \e^{-\int_s^{x_{m+1}} r_{X(u)}\dd u}\cdots \e^{-\int_s^{x_k} r_{X(u)}\dd u} \dd A(x_k)\cdots \dd A(x_{m+1}) \Bigg) \dd N^{ab}(s)   \bigg| X(t)=i \bigg) \nonumber \\
 &=& \sum_{m=2}^k{k \choose m}\int_t^T \Exp \left( \left. 1\{ X(T)=j \} Y \, Z  \dd N^{ab}(s) \right| X(t)=i \right),
 \label{app:proof3} 
\end{eqnarray}
where $Y=\e^{-k \int_t^{s} r_{X(u)}\dd u}  b^{ab}(s)^m $ and
\[  Z=  \int_{s}^T \cdots \int_s^T  \e^{-\int_s^{x_{m+1}} r_{X(u)}\dd u}\cdots \e^{-\int_s^{x_k} r_{X(u)}\dd u} \dd A(x_k)\cdots \dd A(x_{m+1}) .\]
Further conditioning on a lump sum triggering event at time $s$, caused by $\{ X(s)=b, X(s-)=a\}$ and the probability of which is  $d_{ab}(s)\dd s$, and using that $1\{ X(T)=j\}Z$ and $1\{ X(s-)=a\}Y$ are conditionally independent given 
$X(s)=b$, \eqref{app:proof3} reduces to
\begin{eqnarray}
&&\sum_{m=2}^k{k \choose m}\int_t^T  \Exp \Bigg( 1\{ X(T)=j \} Z \bigg| X(s)=b \Bigg) \Exp \bigg(1\{ X(s)=a\} Y  \bigg| X(t)=i \bigg)d_{ab}(s)\dd s\nonumber \\
&=& \sum_{m=2}^k{k \choose m}\int_t^T D^{(k)}_{ia}(t,s)d_{ab}(s) b^{ab}(s)^m V^{(k-m)}_{bj}(s,T)\dd s . \label{app:proof4}
\end{eqnarray}
Summing over $a$ and $b$, and putting \eqref{app:proof4} on matrix form (in $i,j$) this amounts to
 \begin{equation}
   \sum_{m=2}^k{k \choose m}\int_t^T \mat{D}^{(k)}(t,s)\mat{C}^{(m)}(s) \mat{V}^{(k-m)}(s,T)\dd s . \label{app:proof5}
\end{equation}

Now we consider the integral when there are no coincidences. To this end, we rewrite
\begin{eqnarray*}
   \left( \int_t^T\e^{-\int_t^x r_{X(u)}\dd u} \dd A(x)  \right)^k &=& k\int_t^T \e^{-\int_t^x r_{X(u)}\dd u}  \left( \int_x^T \e^{-\int_t^y r_{X(u)}\dd u} \dd A(y)  \right)^{k-1} \dd A(x) \\
   &=& k\int_t^T \e^{-k\int_t^x r_{X(u)}\dd u}  \left( \int_x^T \e^{- \int_x^y r_{X(u)}\dd u} \dd A(y)  \right)^{k-1} \dd A(x)
\end{eqnarray*}

Then
\begin{eqnarray}
\lefteqn{V_{ij}^{(k)}(t,T)=\Exp \left. \left( 1\{ X(T)=j \} \left( \int_t^T\e^{-\int_t^x r_{X(u)}\dd u} \dd A(x)  \right)^k \right| X(t)=i \right)}~~ \nonumber \\ 
&=& k \int_t^T  \Exp \left. \left( 1\{ X(T)=j \} \e^{-k\int_t^x r_{X(u)}\dd u}   \left( \int_x^T\e^{-\int_x^y r_{X(u)}\dd u} \dd A(y)  \right)^{k-1} \dd A(x)  \right| X(t)=i \right)\nonumber \\
&=&k \sum_{\ell} \int_t^T  \Exp \bigg(  1\{ X(x)=\ell \}  \e^{-k\int_t^x r_{X(u)}\dd u}1\{ X(T)=j \}\nonumber \\
&&\times   \Exp \bigg( 1\{ X(T)=j \} \left( \int_x^T\e^{-\int_x^y r_{X(u)}\dd u} \dd A(y)  \right)^{k-1} \dd A(x) \bigg| X(x)=\ell \bigg)  \bigg| X(t)=i \bigg) .\label{app:proof6}
\end{eqnarray}
On the event that $X(x)=\ell$, the contribution to the expectation of the above integral \eqref{app:proof6}, where no coincidences are allowed (i.e., the reward at time $x$ from at most one jump and benefit rates), amounts to
\begin{eqnarray*}
 \lefteqn{\Exp \bigg( 1\{ X(T)=j \} \left( \int_x^T\e^{-\int_x^y r_{X(u)}\dd u} \dd A(y)  \right)^{k-1} \bigg| X(x)=\ell \bigg) b^\ell (x)\dd x }~~~\\
 && + \sum_{m} \Exp \bigg( 1\{ X(T)=j \} \left( \int_x^T\e^{-\int_x^y r_{X(u)}\dd u} \dd A(y)  \right)^{k-1} \bigg| X(x)=m \bigg) d_{\ell m}(x) b^{\ell m} (x)\dd x \\
 &=& b^{\ell}(x) V_{\ell j}^{(k-1)}(x,T)\dd x + \sum_m d_{\ell m}(x)b^{\ell m}(x) V_{mj}^{(k-1)}(x,T)\dd x,
\end{eqnarray*}
and the integral \eqref{app:proof6} then equals
\begin{eqnarray*}
\lefteqn{k \sum_{\ell} \int_t^T  \Exp \bigg(  1\{ X(x)=\ell \}  \e^{-k\int_t^x r_{X(u)}\dd u} \bigg| X(t)=i\bigg)}~~ \\
&& \times \bigg( b^{\ell}(x) V^{(k-1)}_{\ell j}(x,T) + \sum_m d_{\ell m}(x)b^{\ell m}(x) V_{mj}^{(k-1)}(x,T)\bigg) \dd x
\end{eqnarray*}
which in matrix form amounts to
\begin{equation}
 k \int_t^T \mat{D}^{(k)}(t,x)\mat{R}(x) \mat{V}^{(k-1)}(x,T)\dd x .\label{app:proof7} \end{equation}
Adding \eqref{app:proof5} and \eqref{app:proof7} then prove the result.
\end{proof}

\bibliographystyle{plainnat}

\end{document}